\documentclass[11pt]{article}
\usepackage{url,amsfonts, amsmath, amssymb, amsthm}
\usepackage{thmtools}
\usepackage{thm-restate}
\usepackage[T1]{fontenc}
\usepackage[pdftex]{graphicx}
\usepackage{color,soul}
\usepackage{enumitem}
\newcommand{\shortbar}{\begin{center}\rule{5ex}{0.1pt}\end{center}}

\theoremstyle{plain}
\newtheorem{theorem}{Theorem}[section]
\newtheorem{lemma}[theorem]{Lemma}
\newtheorem{corollary}[theorem]{Corollary}

\theoremstyle{definition}

\newtheorem{criterion}{Criterion}
\newtheorem{property}{Property}

\theoremstyle{remark}

\newcommand{\Madry}{M\k{a}dry}

\newcommand{\paren}[1]{\left( #1 \right)}

\newcommand{\ceil}[1]{\lceil #1 \rceil}
\newcommand{\floor}[1]{\lfloor #1 \rfloor}
\newcommand{\f}[2]{\frac{#1}{#2}}
\newcommand{\fr}[2]{\mbox{$\frac{#1}{#2}$}}

\newcommand{\bydef}{\stackrel{\operatorname{def}}{=}}
\newcommand{\poly}{\operatorname{poly}}

\newcommand{\bottom}{\perp}
\newcommand{\rt}{\operatorname{rt}}

\newcommand{\ignore}[1]{}

\newcommand{\rb}[2]{\raisebox{#1 mm}[0mm][0mm]{#2}}
\newcommand{\istrut}[2][0]{\rule[- #1 mm]{0mm}{#1 mm}\rule{0mm}{#2 mm}}
\newcommand{\zero}[1]{\makebox[0mm][l]{$#1$}}

\newcommand{\vcm}[1][1]{\vspace*{#1 cm}}
\newcommand{\hcm}[1][1]{\hspace*{#1 cm}}

\newcommand{\INT}{{\sc integer weights}}
\newcommand{\RAND}{{\sc randomized}}
\newcommand{\MCM}{{\sc mcm}}

\newcommand{\MWPM}{{\sc mwpm}}
\newcommand{\MWPMs}{{\sc mwpm}s}
\newcommand{\Eelig}{E_{\operatorname{elig}}}
\newcommand{\Gelig}{G_{\operatorname{elig}}}
\newcommand{\GeligC}{\widehat{G}_{\operatorname{elig}}}

\newcommand{\Edmonds}{\operatorname{Edm}}

\newcommand{\EdmondsSearch}{\mbox{\sc EdmondsSearch}}
\newcommand{\PQSearch}{\mbox{\sc PQSearch}}
\newcommand{\SearchOne}{\mbox{\sc SearchOne}}
\newcommand{\BucketSearch}{\mbox{\sc BucketSearch}}
\newcommand{\ShellSearch}{\mbox{\sc ShellSearch}}

\newcommand{\Path}{\mbox{\sc DismantlePath}}

\newcommand{\Liquidationist}{\mbox{\sc Liquidationist}}
\newcommand{\Hybrid}{\mbox{\sc Hybrid}}
\newcommand{\Liquidate}{\mbox{\sc Liquidate}}

\newcommand{\now}{\operatorname{now}}
\newcommand{\Outer}{\operatorname{out}}
\newcommand{\Inner}{\operatorname{in}}
\newcommand{\Vin}{V_{\Inner}}
\newcommand{\Vout}{V_{\Outer}}
\newcommand{\VinC}{\widehat{V}_{\Inner}}
\newcommand{\VoutC}{\widehat{V}_{\Outer}}
\renewcommand{\root}{\operatorname{root}}
\newcommand{\new}{\operatorname{new}}
\newcommand{\slack}{\operatorname{slack}}
\newcommand{\key}{\operatorname{key}}
\newcommand{\UFaddedge}{\mathsf{addedge}}
\newcommand{\UFunite}{\mathsf{unite}}
\newcommand{\UFfind}{\mathsf{find}}
\newcommand{\SFinit}{\mathsf{init}}
\newcommand{\SFlist}{\mathsf{list}}
\newcommand{\SFsplit}{\mathsf{split}}
\newcommand{\SFdeckey}{\mathsf{decreasekey}}
\newcommand{\SFfindmin}{\mathsf{findmin}}

\newcommand{\grow}{\mathsf{grow}}
\newcommand{\dissolve}{\mathsf{dissolve}}
\newcommand{\blossom}{\mathsf{blossom}}
\newcommand{\schedule}{\mathsf{schedule}}

\begin{document}

\title{Scaling Algorithms for Weighted Matching in General Graphs\thanks{An extended abstract of this work
was presented in Barcelona, Spain at the Twenty-eighth Annual ACM-SIAM Symposium on Discrete Algorithms (SODA 2017).
Supported by NSF grants CCF-1217338, CNS-1318294, CCF-1514383, CCF-1637546, and BIO-1455983, and AFOSR Grant FA9550-13-1-0042.
R. Duan is supported by a China Youth 1000-Talent grant.
Email: {\tt duanran@mail.tsinghua.edu.cn}, {\tt pettie@umich.edu}, {\tt hsinhao@mit.edu}.}}

\author{Ran Duan\\Tsinghua Univ. \and Seth Pettie \\ Univ. of Michigan \and Hsin-Hao Su \\ Univ. of North Carolina, Charlotte}
\date{}

\maketitle

\begin{abstract}
We present a new scaling algorithm for maximum (or minimum) weight perfect matching on general, edge weighted graphs.
Our algorithm runs in $O(m\sqrt{n}\log(nN))$ time, $O(m\sqrt{n})$ per scale, which matches the running time of the best
cardinality matching algorithms on sparse graphs~\cite{Vazirani12,Vazirani14,GT91,Gabow17}.  
Here $m,n,$ and $N$ bound the number of edges, vertices, and magnitude of any integer edge weight.
Our result improves on a 25-year old algorithm 
of Gabow and Tarjan, which runs in 
$O(m\sqrt{n\log n\alpha(m,n)} \log(nN))$ time.
\end{abstract}

\section{Introduction}

In 1965 Edmonds~\cite{Edmonds65,Ed65} proposed the complexity class $\mathbf{P}$ and
proved that on general (non-bipartite) graphs, 
both the maximum cardinality
matching and maximum weight matching problems could be solved in polynomial time.
Subsequent work on general weighted graph matching 
focused on developing faster implementations of Edmonds' 
algorithm~\cite{Lawler76,Gabow76,Karzanov76,GalilMG86,GGS89,G90,Gabow16}
whereas others pursued alternative techniques such as cycle-canceling~\cite{CunninghamM78}, 
weight-scaling~\cite{G85,GT91}, 
or an algebraic approach using fast matrix multiplication~\cite{CyganGS15}.
Refer to Table~\ref{table:history} for a survey of weighted matching algorithms on general graphs.
The fastest implementation of Edmonds' algorithm~\cite{Gabow16} runs in $O(mn + n^2\log n)$ time
on arbitrarily-weighted graphs.  On graphs with integer edge-weights having magnitude at most $N$,
Gabow and Tarjan's~\cite{GT91} algorithm runs in $O(m\sqrt{n\alpha(m,n)\log n}\log(nN))$ time
whereas Cygan, Gabow, and Sankowski's runs in 
$O(Nn^{\omega})$ time with high probability, where $\omega$ is the matrix multiplication exponent.
For reasonable values of $m,n,$ and $N$ the Gabow-Tarjan algorithm is theoretically superior
to the others.  However, it is an $\Omega(\sqrt{\log n\alpha(m,n)})$ factor slower than comparable algorithms for \emph{bipartite} graphs~\cite{GT89,OrlinA92,GoldbergK97,DuanS12}, 
and even slower than the interior point algorithm of~\cite{CohenMSV17} for \emph{sparse bipartite} graphs.
Moreover, its analysis is rather complex.

In this paper we present a new {\em scaling} algorithm for weighted matching on general graphs
that runs in $O(m\sqrt{n}\log(nN))$ time.   
Each scale of our algorithm runs in $O(m\sqrt{n})$ time, which is asymptotically
the same time required to compute a maximum cardinality matching in a 
sparse graph~\cite{Vazirani12,Vazirani14,GT91,Gabow17}.  
Therefore, 
it is unlikely that our algorithm could be substantially improved without first finding a faster algorithm for the manifestly simpler
problem of cardinality matching.
Our algorithm's time bound also matches that of the best \emph{bipartite} scaling algorithms~\cite{GT89,OrlinA92,GoldbergK97,DuanS12},
but is still slower than~\cite{CohenMSV17} on sufficiently sparse bipartite graphs.

\begin{table}[t]
\centerline{\small
\begin{tabular}{|c|l|l|@{\istrut[1]{3.5}}}
\multicolumn{3}{c}{}\\
\multicolumn{1}{l}{Year} & \multicolumn{1}{l}{Authors}	& \multicolumn{1}{l}{Time Complexity \& Notes}\\\cline{1-3}
1965 & Edmonds	& $mn^2$	 \\\cline{1-3}
1974 & Gabow	& 	 \\
1976 & Lawler		& \rb{2.5}{$n^3$}		\\\cline{1-3}
1976 & Karzanov	& $n^3 + mn\log n$		\\\cline{1-3}
1978 & Cunningham \& Marsh & $\poly(n)$		\\\cline{1-3}
1982 & Galil, Micali \& Gabow & $mn\log n$    \\\cline{1-3}
1985 & Gabow	 &  	$mn^{3/4}\log N$ \hfill \INT\\\cline{1-3}
1989 & Gabow, Galil \& Spencer \hcm[.1] & $mn\log\log\log_d n + n^2\log n$ \hfill $d=2+m/n$\\\cline{1-3}
1990 & Gabow	& $mn + n^2\log n$	 \\\cline{1-3}
1991 & Gabow \& Tarjan\istrut[2]{4}	& $m\sqrt{n\alpha(n,m)\log n}\log(nN)$\hcm  \hfill \INT\\\cline{1-3}
	& Cygan, Gabow &\\
\rb{2.5}{2012} & \& Sankowski & \rb{2.5}{$Nn^\omega$} \hfill \rb{2.5}{\RAND, \INT}\\\cline{1-3}\cline{1-3}
\multicolumn{2}{|c|}{\bf \rb{-2.5}{new}}    & \istrut[2]{4}$\Edmonds \cdot \sqrt{n}\log(nN)$ \hfill \rb{-2.5}{\INT}\\
\multicolumn{2}{|c|}{ }		& \istrut[2]{0}$m\sqrt{n}\log(nN)$\\\cline{1-3}
\end{tabular}
}
\caption{\label{table:history} Maximum Weight Perfect Matching (\MWPM) algorithms for 
General Graphs.  $\Edmonds$ is the time for one execution of Edmonds' search on an integer-weighted graph.}
\end{table}

\subsection{Terminology}\label{sect:terminology}

The input is a graph $G=(V,E,\hat{w})$ where $|V|=n,|E|=m$, and $\hat{w} : E\rightarrow\mathbb{R}$ assigns a real weight to each edge.
A \emph{matching} $M$ is a set of vertex-disjoint edges.  A vertex is \emph{free} if it is not adjacent to an $M$ edge. An \emph{alternating path} is one whose edges alternate between $M$ and $E\setminus M$. An alternating path $P$ is \emph{augmenting} if it begins and ends with free vertices, which implies that $M\oplus P \bydef (M\cup P)\setminus (M\cap P)$ is also a matching and has one more edge.  
The {\em maximum cardinality matching} (\MCM) problem is to find a matching $M$ maximizing $|M|$.
The {\em maximum weight perfect matching} (\MWPM) problem is to find a perfect matching $M$ (or, in general, one with maximum cardinality) maximizing $\hat{w}(M) = \sum_{e\in M}\hat{w}(e)$.
The {\em maximum weight matching} problem (with no cardinality constraint) is reducible to \MWPM~\cite{DuanP14} and may be
a slightly easier problem~\cite{DuanS12,HuangK12}.
In this paper we assume that $\hat{w} : E\rightarrow \{0,\ldots,N\}$ assigns non-negative integer weights bounded by $N$.\footnote{Assuming non-negative weights is without loss of generality since we can simply subtract 
$\min_{e\in E}\{\hat{w}(e)\}$ from every edge weight, which does not affect the relative weight of two perfect matchings.
Moreover, the {\em minimum} weight perfect matching problem is reducible to \MWPM, simply by substituting $-\hat{w}$ for $\hat{w}$.}

\subsection{Edmonds' Algorithm}
Edmonds' \MWPM{} algorithm begins with an empty matching 
$M$ and consists of a sequence of {\em search} steps, 
each of which performs zero or more {\em dual adjustment}, {\em blossom shrinking}, and {\em blossom dissolution} steps
until a tight augmenting path emerges or the search detects that $|M|$ is 
maximum.  (Blossoms, duals, and tightness are reviewed in Section~\ref{sect:historymatching}.)
The overall running time is therefore $O(n \cdot \Edmonds)$, where 
$\Edmonds$ is the cost of one search. 
Gabow's implementation~\cite{Gabow16} of Edmonds' search runs in $O(m+n\log n)$ time,
the same as one Hungarian search~\cite{FT87} on bipartite graphs.

\subsection{Scaling Algorithms}\label{sect:scalingalgorithms}

The problem with Edmonds' \MWPM{} algorithm is that it finds augmenting paths one at a time, apparently dooming it to a running time of $\Omega(mn)$.
The matching algorithms of~\cite{G85,GT91} take the
{\em scaling} approach of Edmonds and Karp~\cite{EdmondsK72}.
The idea is to expose the edge weights one bit at a time.  
In the $i$th scale the goal is to compute an optimum perfect matching with respect 
to the $i$ most significant bits of $\hat{w}$.  Gabow~\cite{G85}
showed that each of $\log N$ scales can be solved in $O(mn^{3/4})$ time. 
Gabow and Tarjan~\cite{GT91} observed that it suffices to compute a 
$\pm O(n)$-approximate solution at each scale, provided there are additional scales; 
each of their $\log(nN)$ scales can be solved in $O(m\sqrt{n\alpha(m,n)\log n})$ time.

Scaling algorithms for general graph matching face a unique difficulty not encountered
by scaling algorithms for other optimization problems.
At the beginning of the $i$th scale
we have inherited from the $(i-1)$th scale a nested set $\Omega'$ of blossoms
and near-optimal duals $y',z'$.  (The matching primer in Section~\ref{sect:historymatching} reviews $y$ and $z$ duals.)  
Although $y',z'$ are {\em numerically} close to optimal, 
$\Omega'$ may be {\em structurally} very far from optimal for scale $i$.  
The~\cite{G85,GT91} algorithms {\em gradually} get rid of inherited blossoms in $\Omega'$
while simultaneously building up a new near-optimum solution $\Omega,y,z$.
They decompose
the tree of $\Omega'$ blossoms into heavy paths and process the paths in a bottom-up fashion.
Whereas Gabow's method~\cite{G85} is slow but moves the dual objective in the right direction,
the Gabow-Tarjan method~\cite{GT91} is faster but may actually {\em widen} the gap between the
dual objective and optimum.  There are $\log n$ layers of heavy paths and processing each layer
widens the gap by up to $O(n)$.  Thus, at the final layer the gap is $O(n\log n)$.
It is this gap that is the source of the $\sqrt{n\log n}$ factor in the running time of~\cite{GT91},
not any data structuring issues.

Broadly speaking, our algorithm follows the scaling approach of~\cite{G85,GT91},
but dismantles old blossoms in a completely new way, and further weakens the goal
of each scale.  Rather than compute an optimal~\cite{G85} or near-optimal~\cite{GT91}
perfect matching at each scale, we compute a near-optimal, {\em near-perfect} 
matching at each scale.  The advantage of leaving some vertices unmatched
(or, equivalently, artificially matching them up with dummy mates) is not at all obvious, 
but it helps speed up the dismantling of blossoms in the {\em next} scale. 
The algorithms are parameterized by a $\tau = \tau(n)$.  A blossom is called {\em large} if it contains 
at least $\tau$ vertices and {\em small} otherwise.  Each scale of our algorithm
produces an {\em im}perfect matching $M$ with $y,z,\Omega$ that 
(i) leaves $O(n/\tau)$ vertices unmatched, and 
(ii) is such that the sum of $z(B)$ of all large $B\in\Omega$ is $O(n)$,
     independent of the magnitude of edge weights.
After the last scale, the vertices left free by (i) will need to be matched up in $O(\Edmonds\cdot (n/\tau))$ time,
at the cost of one Edmonds' search per vertex.  Thus, we want $\tau$ to be large.
Part (ii) guarantees that large blossoms formed in one scale can be efficiently {\em liquidated} 
in the next scale (see Section~\ref{sect:liquidationist}),
but getting rid of small blossoms (whose $z$-values are unbounded, as a function of $n$)
is more complicated.  Our methods for getting rid of small blossoms have running 
times that are increasing with $\tau$, so we want $\tau$ to be small.
In the \Liquidationist{} algorithm, all inherited small blossoms are processed in 
$O(\Edmonds\cdot \tau)$ time whereas in \Hybrid{} (a hybrid of \Liquidationist{}
and Gabow's algorithm~\cite{G85}) they are processed in $O(m\tau^{3/4})$ time.

\subsection{Organization}
In Section~\ref{sect:historymatching} we review Edmonds' LP formulation of \MWPM{}
and Edmonds' search procedure.
In Section~\ref{sect:liquidationist} we present the \Liquidationist{} algorithm
running in $O(\Edmonds \cdot \sqrt{n}\log(nN))$ time.
In Section~\ref{sect:hybrid} we give the \Hybrid{} algorithm running in 
$O(m\sqrt{n}\log(nN))$ time.

Our algorithms depend on a having an efficient implementation
of Edmonds' search procedure.  
In Section~\ref{sect:ImplementingEdmonds} we give a detailed
description of an implementation of Edmonds' search that is very efficient
on integer-weighted graphs.  It runs in linear time when there are a linear number
of dual adjustments.  When the number of dual adjustments is unbounded it runs in 
$O(m\log\log n)$ time deterministically or $O(m\sqrt{\log\log n})$ time w.h.p.
This implementation is based on ideas suggested by Gabow~\cite{G85} and 
may be considered folklore in some quarters.

We conclude with some open problems in Section~\ref{sect:conclusion}.

\section{A Matching Primer}\label{sect:historymatching}

The \MWPM{} problem can be expressed as an \underline{integer} linear program
\begin{align*}
\mbox{maximize } \;\; & \sum_{e\in E} x(e)\cdot \hat{w}(e)\\
\mbox{subject to } \;\; & x(e) \in \{0,1\}, \mbox{ for all $e\in E$}\\
\mbox{ and }	\;\; & \sum_{e \ni v} x(e) = 1, \mbox{ for all $v\in V$.}
\end{align*}
The integrality constraint lets us interpret $\mathbf{x}$ as the membership vector of a set of edges
and the $\sum_{e\ni v} x(e)=1$ constraint enforces that $\mathbf{x}$ represents a perfect matching.
Birkhoff's theorem~\cite{Birkhoff46} (see also von Neumann~\cite{vonNeumann53})
implies that in bipartite graphs
the integrality constraint can be relaxed to $x(e) \in [0,1]$.  The basic feasible solutions to the resulting LP
correspond to perfect matchings.  However, this is not true of non-bipartite graphs!
Edmonds proposed exchanging the integrality constraint for an exponential
number of the following {\em odd set} constraints, 
which are obviously satisfied for every $\mathbf{x}$ that is the membership vector of a matching.
\begin{align*}
& \sum_{e \in E(B)} x(e) \le \floor{|B|/2}, \mbox{ for all $B\subset V$, $|B| \ge 3$ odd.}
\end{align*}
Edmonds proved that the basic feasible solutions to the resulting LP are integral and therefore correspond
to perfect matchings.  Weighted matching algorithms work directly with the dual LP.
Let $y : V\rightarrow \mathbb{R}$ and $z : 2^{V} \rightarrow \mathbb{R}$ be the vertex duals and odd set duals.
\begin{align*}
\mbox{minimize } \;\; & \sum_{v\in V} y(v) + \sum_{\substack{B \subset V :\\ |B|\ge 3 \, \operatorname{is}\,\operatorname{odd}}} z(B)\cdot \floor{|B|/2}\\
\mbox{subject to } \;\; & z(B) \ge 0, \mbox{ for all odd $B\subset V$,}\\
			& \hat{w}(u,v) \le yz(u,v) \mbox{ for all $(u,v) \in E$,}\\
\mbox{where, by definition, } \;\; & yz(u,v) \bydef y(u) + y(v) + \sum_{B\supset \{u,v\}} z(B).
\end{align*}
We generalize the synthetic dual $yz$ to an arbitrary set $S\subseteq V$ of vertices as follows.
\[
yz(S) = \sum_{u \in S} y(u) \,+\, \sum_{B \subseteq S}z(B) \cdot \floor{|B|/2} \,+\, \sum_{B\supset S}z(B) \cdot \floor{|S|/2}.
\]
Note that $yz(V)$ is exactly the dual objective.

Edmonds' algorithm~\cite{Edmonds65,Ed65} maintains a dynamic matching $M$ and dynamic 
laminar set $\Omega \subset 2^V$ of odd sets, each associated with a {\em blossom} subgraph.
Informally, a blossom is an odd-length alternating cycle (w.r.t.~$M$), whose constituents 
are either individual vertices or blossoms in their own right.  
More formally, blossoms are constructed inductively as follows.
If $v\in V$ then the odd set
$\{v\}$ induces a trivial blossom with edge set $\emptyset$.
Suppose that for some odd $\ell\ge 3$, 
$A_0,\ldots,A_{\ell-1}$ are disjoint sets associated with blossoms
$E_{A_0},\ldots,E_{A_{\ell-1}}$.  If there are edges $e_0,\ldots,e_{\ell-1}\in E$
such that $e_i \in A_i \times A_{i+1}$ (modulo $\ell$) 
and $e_i\in M$ if and only if $i$ is odd, then $B = \bigcup_i A_i$ is an odd set
associated with the blossom $E_B = \bigcup_i E_{A_i} \cup \{e_0,\ldots,e_{\ell-1}\}$.
Because $\ell$ is odd, the alternating cycle on $A_0,\ldots,A_{\ell-1}$ 
has odd length, leaving $A_0$ incident to two unmatched edges, $e_0$ and $e_{\ell-1}$.
One can easily prove by induction that $|B|$ is odd and that $E_B \cap M$ 
matches all but one vertex in $B$, called the {\em base} of $B$. 
Remember that $E(B) = E\cap {B\choose 2}$,\footnote{The notation $X\choose t$ refers to the set of all subsets of $X$ of size $t$,
so $B\choose 2$ is the set of all possible undirected edges on $B$.} 
the edge set induced by $B$, may contain many
non-blossom edges not in $E_B$.  Define $n(B) = |B|$ and $m(B) = |E(B)|$ to be the number of vertices
and edges in the graph induced by $B$.

The set $\Omega$ of {\em active} blossoms is represented by rooted trees,
where leaves represent vertices and internal nodes represent nontrivial blossoms.
A {\em root blossom} is one not contained in any other blossom.  The children of an internal node representing a blossom $B$ are
ordered by the odd cycle that formed $B$, where the child containing
the base of $B$ is ordered first.
Edmonds~\cite{Ed65,Edmonds65} showed 
that it is often possible to treat blossoms as if they were single vertices, by {\em shrinking} them.
We obtain the {\em shrunken graph} $G/\Omega$ by contracting 
all root blossoms and removing the edges in those blossoms.
To {\em dissolve} a root blossom $B$ means
to delete its node in the blossom forest and, in the contracted graph, to replace 
$B$ with individual vertices $A_0,\ldots,A_{\ell-1}$.

\begin{figure}
\centering
\begin{tabular}{c@{\hcm[1.7]}c}
\scalebox{.28}{\includegraphics{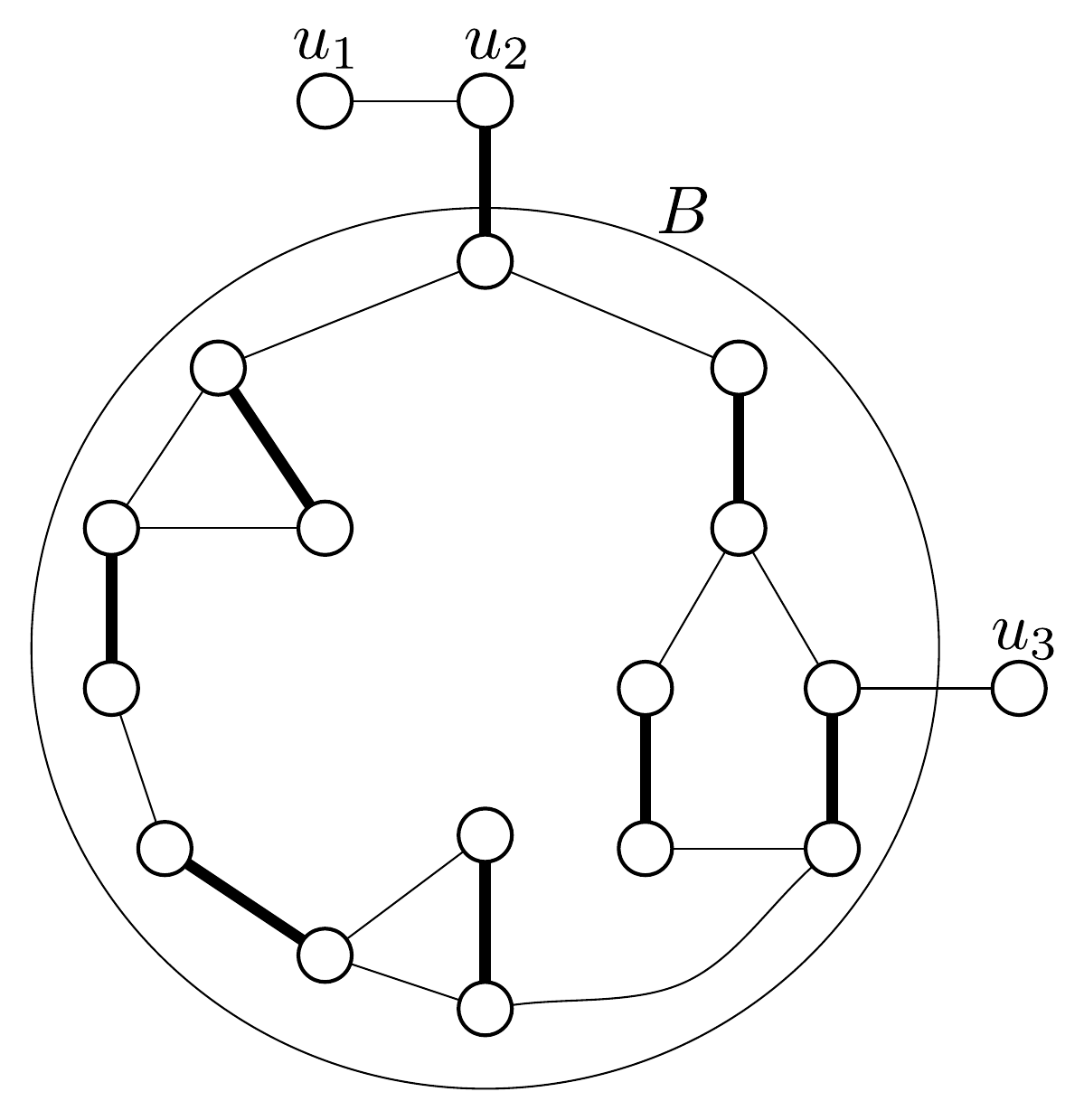}}
&
\scalebox{.28}{\includegraphics{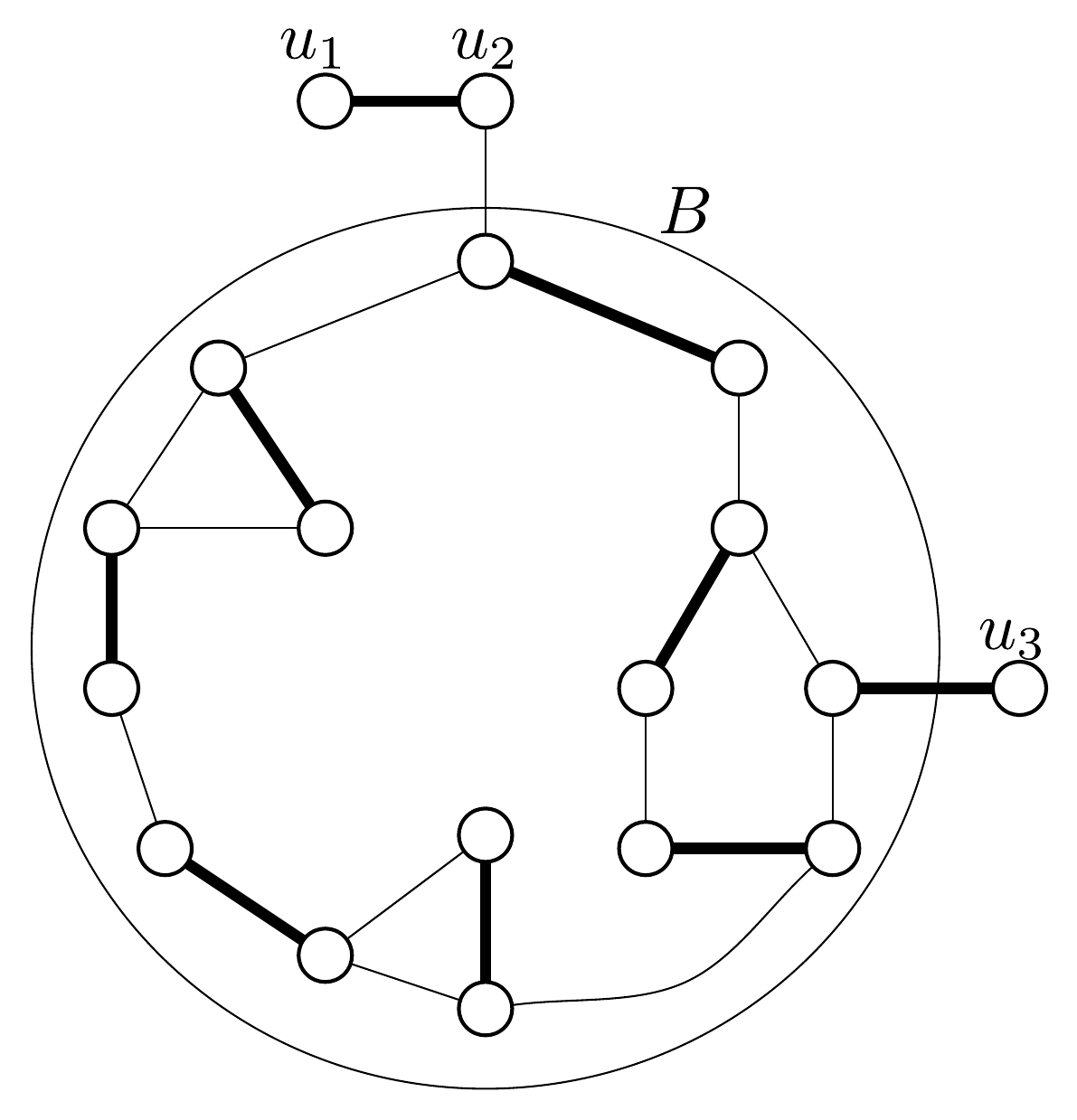}}\\
(a) & (b)
\end{tabular}
{\scriptsize 
\caption{\label{fig:shrunken}
Matched edges are think, unmatched edges thin.  
Left: $B$ is a blossom consisting of 7 sub-blossoms, 4 of which are trivial (vertices) and the other
three non-trivial blossoms.  The path $P' = (u_1,u_2,B,u_3)$ is an augmenting path in 
the shrunken graph $G/\{B\}$.  Right: augmenting along $P'$ in $G/\{B\}$ enlarges the matching and 
has the effect of moving the base of $B$ to the vertex matched with $u_3$.}}
\end{figure}

Blossoms have numerous properties.  Our algorithms use two in particular.

\begin{enumerate}[topsep=0pt,itemsep=0ex,partopsep=1ex,parsep=1ex]
\item The subgraph on $E_B$ is {\em critical}, meaning it contains a perfect matching on $B\backslash\{v\}$, for each $v\in B$.
Phrased differently, any $v\in B$ can be made the base of $B$ by choosing the matching edges in $E_B$ appropriately.

\item As a consequence of (1), any augmenting path $P'$ in $G/\Omega$ extends to an augmenting path $P$ in $G$,
by replacing each non-trivial blossom vertex $B$ in $P'$ with a corresponding path through $E_B$.  Moreover, $\Omega$ 
is still valid for the matching $M\oplus P$, though the bases of blossoms intersecting $P$ may be relocated by augmenting along $P$.
See Figure~\ref{fig:shrunken} for an example.
\end{enumerate}

\subsection{Relaxed Complementary Slackness}

Edmonds' algorithm maintains a matching $M$, a nested set $\Omega$ of blossoms, 
and duals $y:V\rightarrow\mathbb{Z}$ and $z:2^{V}\rightarrow\mathbb{N}$ that satisfy Property~\ref{prop:CS}.
Here $w$ is a weight function assigning {\em even} integers; it is generally not the same as the input weights $\hat{w}$.

\begin{property}(Complementary Slackness)\label{prop:CS} Assume $w$ assigns only {\em even} integers.
\begin{enumerate}[topsep=0.5ex,itemsep=-.5ex]
	\item {\it Granularity.} $z(B)$ is a nonnegative even integer and $y(u)$ is an integer.
	\item {\it Active Blossoms.} $|M\cap E_B| = \floor{|B|/2}$ for all $B\in \Omega$.  If $B \in \Omega$ is a root blossom then $z(B)>0$; if $B \notin \Omega$ then $z(B) = 0$.  Non-root blossoms may have zero $z$-values.
    	\item {\it Domination.} $yz(e)\geq w(e)$, for each $e=(u,v)\in E$.
    	\item {\it Tightness.} $yz(e)= w(e)$, for each $e \in M \cup \bigcup_{B\in \Omega} E_B$.
\end{enumerate}
\end{property}

\begin{lemma}\label{lem:optimality}
If Property~\ref{prop:CS} is satisfied for a perfect matching $M$, blossom set $\Omega$, and duals $y,z$, 
then $M$ is necessarily a \MWPM{} w.r.t.~the weight function $w$.
\end{lemma}

The proof of Lemma~\ref{lem:optimality} follows the same lines as Lemma~\ref{thm:approx}, proved below.
The Gabow-Tarjan algorithms and their successors~\cite{GT89,GT91,OrlinA92,GoldbergK97,DuanS12,DuanP14} maintain a relaxation of complementary slackness.  By using Property~\ref{prop:RCS} in lieu of Property~\ref{prop:CS} we introduce
an additive error as large as $n$.  This does not prevent us from computing exact \MWPMs{} but it does necessitate additional scales.
Before the algorithm proper begins we compute the extended weight function $\bar{w}(e) = (\fr{n}{2}+1)\hat{w}(e)$.  
Note that the weight of every matching w.r.t.~$\bar{w}$ is a multiple of $\fr{n}{2}+1$.
After the final scale of our algorithms $w = 2\bar{w}$ so if we can find a matching $M$ such that
$w(M)$ is additively within $n$ of the \MWPM, then $\bar{w}(M)$ is additively within 
$\fr{n}{2}$ of the \MWPM, which implies that $M$ is exactly optimum w.r.t.~both $\bar{w}(M)$ and 
$\hat{w}(M)$.  These observations motivate the use of Property~\ref{prop:RCS}.

\begin{property}(Relaxed Complementary Slackness)\label{prop:RCS} 
Assume $w$ assigns only {\em even} integers.
Property \ref{prop:CS}(1,2) holds and (3,4) are replaced with
\begin{enumerate}[topsep=0.5ex,itemsep=-0.5ex]
	\item[3.] {\it Near Domination.} $yz(e)\geq w(e)-2$ for each edge $e=(u,v)\in E$. 
	\item[4.] {\it Near Tightness.} $yz(e) \leq w(e)$, for each $e \in M \cup \bigcup_{B\in \Omega} E_B$.
\end{enumerate}
\end{property}

The ``-2'' in Property~\ref{prop:RCS} is due to the fact that $w$ is always even.\footnote{An equivalent implementation would be to assume that $w$ is merely integral, and maintain the invariant that
$z$ is integral and $y$ half-integral.}

\begin{lemma}\label{thm:approx}
If Property~\ref{prop:RCS} is satisfied for some perfect matching $M,$ blossom set $\Omega$, and duals $y,z$,
then $w(M)\geq w(M^*)-n$, where $M^*$ is an \MWPM~w.r.t.~$w$.
\end{lemma}
\begin{proof}
By Property \ref{prop:RCS} (near tightness and active blossoms), the definition of $yz$, and the perfection of $M$, we have
\begin{align*}
w(M) &\geq \sum_{e\in M} yz(e)\:=\: \sum_{u\in V} y(u)+\sum_{B\in\Omega}z(B)\cdot|M\cap E(B)| \; \zero{\displaystyle \:=\:\sum_{u\in V} y(u)+\sum_{B\in\Omega}z(B)\cdot\floor{|B|/2}.}
\intertext{Since the \MWPM{} $M^*$ puts at most $\floor{|B|/2}$ edges in any blossom $B\in \Omega$,}
w(M^*)&\leq \sum_{e\in M^*} (yz(e) + 2)								& \mbox{Property \ref{prop:RCS} (near domination)}\\
	&=\sum_{u\in V} y(u)+\sum_{B\in\Omega}z(B)\cdot|M^*\cap E(B)|+2|M^*| & \mbox{Defn.~of $yz$}\\
&\leq \sum_{u\in V} y(u)+\sum_{B\in\Omega}z(B)\cdot\lfloor|B|/2\rfloor+n.		& \mbox{$|M^*| = n/2$, Non-negativity of $z$} 
\end{align*}

Therefore, we have $w(M)\geq w(M^*)-n$.
\end{proof}

\subsection{Edmonds' Search}

Suppose we have a matching $M$, blossom set $\Omega$, and duals $y,z$ satisfying Property~\ref{prop:CS} or~\ref{prop:RCS}.
The goal of Edmonds' {\em search} procedure is to manipulate $y,z,$ and $\Omega$ until an {\em eligible} augmenting path emerges.
At this point $|M|$ can be increased by augmenting along such a path (or multiple such paths), which preserves Property~\ref{prop:CS} or \ref{prop:RCS}.
The definition of {\em eligible} needs to be compatible with the governing invariant (Property~\ref{prop:CS} or \ref{prop:RCS}) and
other needs of the algorithm.
In our algorithms we use several implementations of Edmonds' generic search: they differ in their governing invariants, 
definition of eligibility, and data structural details.  For the time being the reader can imagine that 
Property~\ref{prop:CS} is in effect 
and that we use Edmonds' original eligibility criterion~\cite{Edmonds65}.

\begin{criterion}\label{crit1}
An edge $e$ is {\em eligible} if it is {\em tight}, that is, $yz(e) = w(e)$.  
\end{criterion}

Each scale of our algorithms begins with Property~\ref{prop:CS} as the governing invariant but switches 
to Property~\ref{prop:RCS} when all inherited blossoms are gone.  
When Property~\ref{prop:RCS} is in effect we use Criterion~\ref{crit2} if the algorithm aims to find augmenting paths
in batches and Criterion~\ref{crit3} when augmenting paths are found one at a time.  The reason for switching
from Criterion~\ref{crit2} to \ref{crit3} is discussed in more detail in the proof of Lemma~\ref{lem:liquidationist-correct}.

\begin{criterion}\label{crit2}
An edge $e$ is \emph{eligible} if at least one of the following holds.
\begin{enumerate}[topsep=0.5ex,itemsep=-0.5ex]
	\item $e\in E_B$ for some $B\in\Omega$.
	\item $e\notin M$ and $yz(e)=w(e)-2$.
	\item $e\in M$ and $yz(e)=w(e)$.
\end{enumerate}
\end{criterion}

\begin{criterion}\label{crit3}
An edge is {\em eligible} if $yz(e) = w(e)$ or $yz(e) = w(e)-2$.
\end{criterion}

Regardless of which eligibility criterion is used, 
let $\Gelig = (V,\Eelig)$ be the eligible subgraph
and $\GeligC = \Gelig/\Omega$ be obtained from $\Gelig$ 
by contracting all root blossoms.

We consider a slight variant of Edmonds' search that looks for augmenting paths only from 
a specified set $F$ of free vertices in $V$, that is, each augmenting path must have at least one end in $F$ and possibly both.  
(We also use $F$ to denote the corresponding free vertices in $\GeligC$.) 
The search iteratively performs {\em Augmentation}, {\em Blossom Shrinking}, {\em Dual Adjustment}, and {\em Blossom Dissolution} steps, halting after the first Augmentation step that discovers at least one augmenting path.  
We require that the $y$-values of all $F$ vertices have the same parity (even/odd). 
This is needed to keep $y,z$ integral and allow us to perform discrete dual adjustment steps without violating Property~\ref{prop:CS} or~\ref{prop:RCS}. See Figure \ref{fig:edmondssearch} for the pseudocode.

\begin{figure}
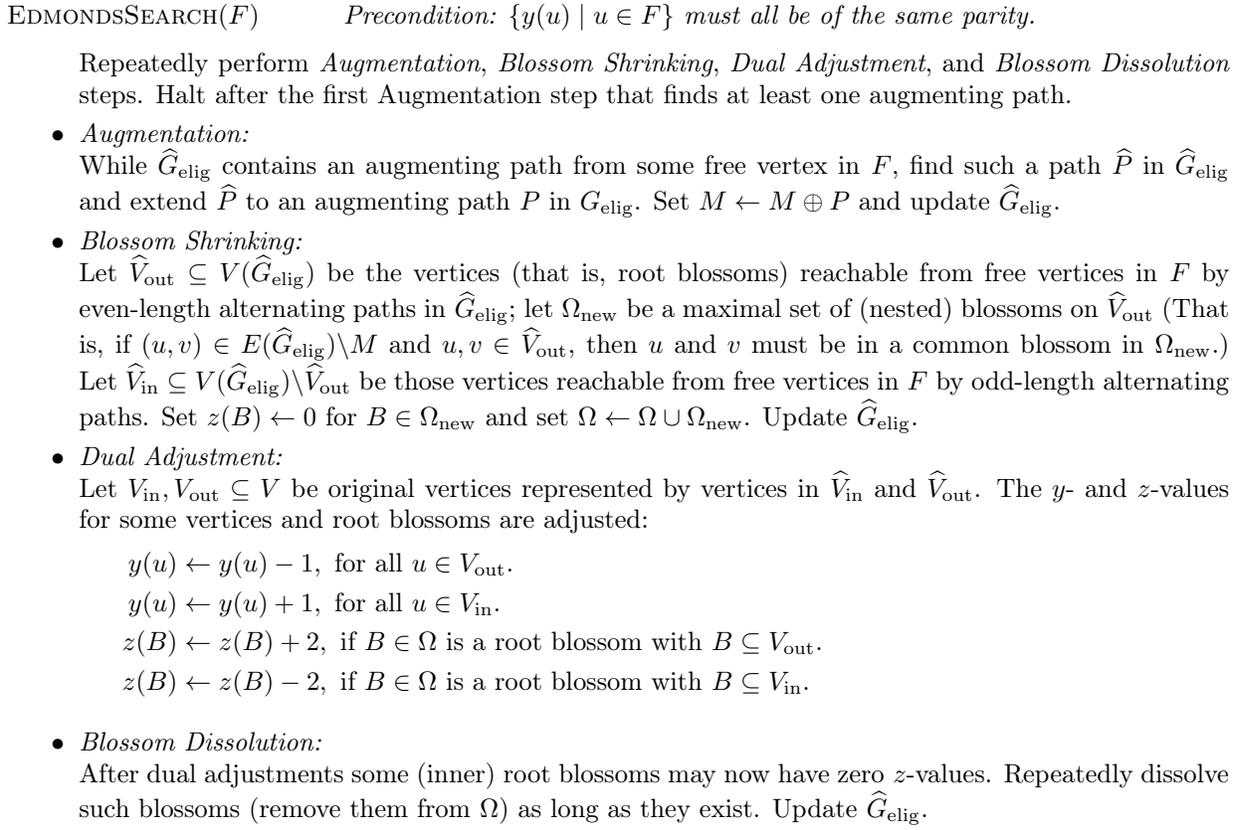

\centering
\framebox{
\begin{minipage}{6.4in}
\small
\noindent $\EdmondsSearch(F)$ \hcm {\em Precondition: $\{y(u) \;|\; u\in F\}$ must all be of the same parity.}
\begin{itemize}[itemsep=0ex]
\item[] Repeatedly perform {\em Augmentation}, {\em Blossom Shrinking}, {\em Dual Adjustment}, and {\em Blossom Dissolution} steps.  Halt after the first Augmentation step that finds at least one augmenting path.

\item {\em Augmentation:}\\  
While $\GeligC$ contains an augmenting path from some free vertex in $F$, find such a path $\widehat{P}$ in $\GeligC$ and 
extend $\widehat{P}$ to an augmenting path $P$ in $\Gelig$. Set $M\leftarrow M\oplus P$ and update $\GeligC$. 

\item \emph{Blossom Shrinking:}\\
    Let $\VoutC \subseteq V(\GeligC)$ be the vertices (that is, root blossoms) reachable from free vertices in $F$ by even-length alternating paths in $\GeligC$;  
    let $\Omega_{\new}$ be a maximal set of (nested) blossoms on $\VoutC$
    (That is, if $(u,v)\in E(\GeligC)\backslash M$ and $u,v\in \VoutC$, then $u$ and $v$ must be in a common blossom in $\Omega_{\new}$.)
    Let $\VinC \subseteq V(\GeligC)\backslash \VoutC$ be those vertices reachable from free vertices in $F$ by odd-length alternating paths.
        Set $z(B) \leftarrow 0$ for $B\in \Omega_{\new}$ and set $\Omega \leftarrow \Omega \cup \Omega_{\new}$.  Update $\GeligC$.

\item \emph{Dual Adjustment:}\\
    Let $\Vin,\Vout \subseteq V$ be original vertices represented by vertices in $\VinC$ and $\VoutC$.
    The $y$- and $z$-values for some vertices and root blossoms are adjusted:
    \vcm[-.2]
    \begin{align*}
        y(u) &\leftarrow y(u)-1, \mbox{ for all $u\in \Vout$.}\\
        y(u) &\leftarrow y(u)+1,  \mbox{ for all $u\in \Vin$.}\\
        z(B) &\leftarrow z(B)+2, \mbox{ if $B\in\Omega$ is a root blossom with $B\subseteq \Vout$.}\\
        z(B) &\leftarrow z(B)-2, \mbox{ if $B\in\Omega$ is a root blossom with $B\subseteq \Vin$.}\hcm[5]
    \end{align*}
    \vcm[-.5]
\item \emph{Blossom Dissolution:}\\
After dual adjustments some (inner) root blossoms may now have zero $z$-values. 
Repeatedly dissolve such blossoms (remove them from $\Omega$) as long as they exist. Update $\GeligC$.
\end{itemize}
\end{minipage}
}
\caption{A \underline{generic} implementation of Edmonds' search procedure.  Data structural issues are ignored,
as is the eligibility criterion, which determines $\GeligC$.}\label{fig:edmondssearch}
\end{figure}

The main data structure needed to implement \EdmondsSearch{} is a priority queue for
scheduling events
(blossom dissolution,
 blossom formation,
and {\em grow} events that add vertices to $V_{\Inner}\cup V_{\Outer}$).
We refer to \PQSearch{} as an implementation of $\EdmondsSearch$ when the 
number of dual adjustments is unbounded.  
See Gabow~\cite{Gabow16} for an implementation of \PQSearch{} taking $O(m+n\log n)$ time,
or Section~\ref{sect:ImplementingEdmonds} for one taking $O(m\sqrt{\log\log n})$ time, w.h.p.
When the number of dual adjustments is $t=O(m)$ we can use a trivial array of buckets as a priority queue.
Let $\BucketSearch$ be an implementation of $\EdmondsSearch$ running in $O(m+t)$ time; refer 
to Section~\ref{sect:ImplementingEdmonds} for a detailed description.

Regardless of what $t$ is or how the dual adjustments are handled, 
we still have options for how to implement the {\em Augmentation} step.
Under Criterion~\ref{crit1} of eligibility, we can make the Augmentation step extend $M$ to a maximum cardinality matching 
in the subgraph of $\Gelig$ induced by $V(M) \cup F$.
This can be done in $O((p+1)m)$ time if $p\ge 0$ augmenting paths are found~\cite{GT85}, 
or in $O(m\sqrt{n})$ time, independent of $p$, using a cardinality matching algorithm, e.g., \cite{MV80,Vazirani12,Vazirani14} or~\cite[\S 10]{GT91} or~\cite{Gabow17}.

When eligibility Criterion~\ref{crit2} is in effect the {\em Augmentation} step is qualitatively different.
Observe that in the contracted graph $G/\Omega$, matched and unmatched edges have {\em different}
eligibility criteria.  It is easily proved that augmenting along a {\em maximal} set of augmenting paths
eliminates all eligible augmenting paths,\footnote{The distinction between a 
maximal set and maximum set of augmenting paths 
is, in the context of flow algorithms, 
entirely analogous to the distinction between blocking flows and maximum flows.}
 quickly paving the way for {\em Blossom Shrinking} and {\em Dual Adjustment}
steps.  Unlike $\PQSearch$ and $\BucketSearch$, $\SearchOne$ only performs {\em one} dual adjustment 
and {\em must} be used with Criterion~\ref{crit2}.
Finding a maximal set of augmenting paths in $O(m)$ time is straightforward with depth first search~\cite[\S 8]{GT91} and a union-find algorithm~\cite{GT85}.

\begin{figure}[h!]
\centering
\framebox{
\begin{minipage}{6.3in}
\small
\noindent $\SearchOne(F)$\hcm {\em Precondition: $\{y(u) \;|\; u\in F\}$ must all be of the same parity.}
\begin{itemize}
\item {\em Augmentation:}\\
Find a maximal set $\Psi$ of vertex-disjoint augmenting paths from $F$ in $\GeligC$.
Set $\displaystyle M \leftarrow M\oplus \bigcup_{P\in \Psi} P$.

\item Perform {\em Blossom Shrinking}, {\em Dual Adjustment}, and {\em Blossom Dissolution} steps from $F$, 
exactly as in $\EdmondsSearch$.
\end{itemize}
\end{minipage}
}
\caption{\label{fig:SearchOne}}
\end{figure}

The following lemmas establish the correctness of $\EdmondsSearch$
(using either Property~\ref{prop:CS} or \ref{prop:RCS})
and $\SearchOne$ (using Property~\ref{prop:RCS} and Criterion~\ref{crit2}).

\begin{lemma}\label{lem:no-eligible-path}
After the Augmentation step of $\SearchOne(F)$ (using Criterion~\ref{crit2} for eligibility), 
$\GeligC$ contains no eligible augmenting paths from an $F$-vertex.
\end{lemma}

\begin{proof}
Suppose that, after the Augmentation step, there is an augmenting path $P$ from an $F$-vertex in $\GeligC$.
Since $\Psi$ was maximal, $P$
must intersect some $P'\in\Psi$ at a vertex $v$.
However, after the Augmentation step every edge in $P'$ will become ineligible, so the
matching edge $(v,v')\in M$ is no longer in $\GeligC$, contradicting the fact that $P$ consists of eligible edges.
\end{proof}

\begin{lemma}\label{lem:offendCS}
If Property \ref{prop:CS} is satisfied and the $y$-values of vertices in $F$ have the same parity, 
then $\EdmondsSearch(F)$ (under Criterion~\ref{crit1}) preserves Property~\ref{prop:CS}.
\end{lemma}

\begin{proof}
Property~\ref{prop:CS} (granularity) is obviously maintained, 
since we are always adjusting $y$-values by 1 and $z$-values by 2.
Property~\ref{prop:CS} (active blossoms) is also maintained since all the new root blossoms discovered in the Blossom Shrinking step are in $V_{\Outer}$ and will have positive $z$-values after adjustment. Furthermore, each root blossom whose $z$-value drops to zero is removed.

Consider the tightness and the domination conditions of Property \ref{prop:CS}. 
First note that if both endpoints of $e$ lie in the same blossom, $yz(e)$ will not change until the blossom is dissolved. 
When the blossom was formed, the blossom edges must be eligible (tight). 
The augmentation step only makes eligible edges matched, so tightness is satisfied. 

Consider the effect of a dual adjustment on an edge $e = (u,v)$, whose endpoints lie in different blossoms. 
We divide the analysis into the following four cases.  Refer to Figure~\ref{fig:dual-adjustment} for illustrations of the cases.

\begin{figure}
\centering
\scalebox{.35}{\includegraphics{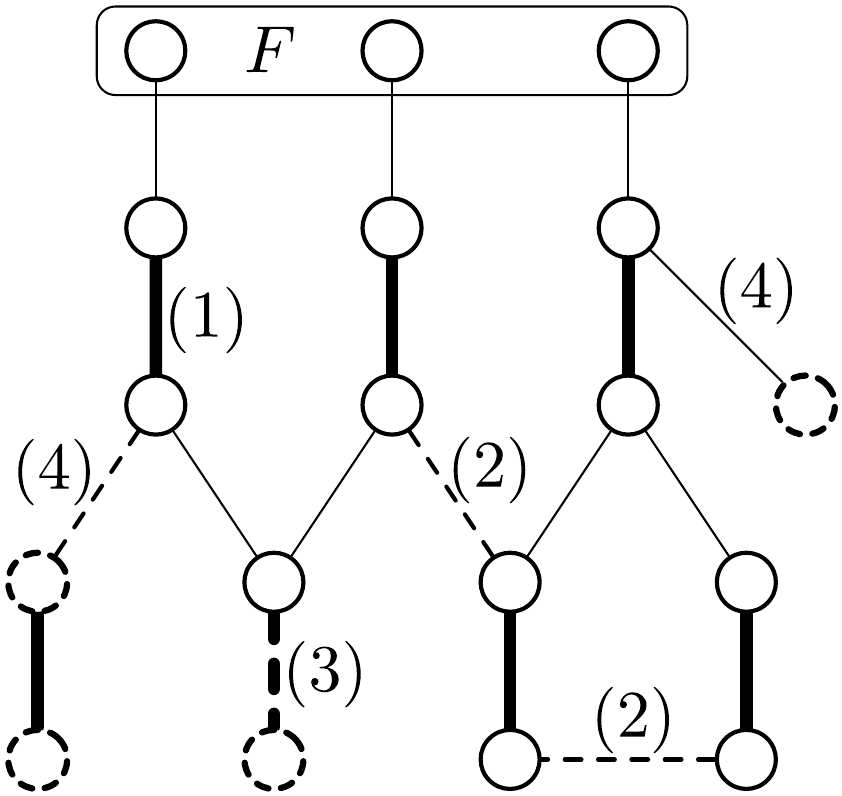}}
\caption{\label{fig:dual-adjustment}Thick edges are matched, thin unmatched.  Dashed edges (whether thick or thin) are ineligible.
Solid vertices are in $\Vin \cup \Vout$; all other vertices are dashed.  Case (3) can only occur under
Criteria~\ref{crit2} or \ref{crit3} of eligibility.}
\end{figure}

\begin{enumerate}
	\item Both $u$ and $v$ are in $\Vin \cup \Vout$ and $e \in M$.  We cannot have both $u,v\in \Vout$ 
	(otherwise they would be in a common blossom, since $e$ is eligible) nor can both be in $\Vin$, so $u\in \Vin, v\in \Vout$
	and $yz(e)$ is unchanged.
	
	\item Both $u$ and $v$ are in $\Vin \cup \Vout$ and $e \notin M$.  If at least one of $u$ or $v$ is in $\Vin$, then $yz(e)$ cannot decrease and domination holds. Otherwise we must have $u,v \in \Vout$. In this case, $e$ must be ineligible, for otherwise an augmenting path or a blossom would have been found.  Ineligibility implies $yz(e) \ge w(e) + 1$ but something stronger can be inferred.
Since the $y$-values of free vertices have the same parity, all vertices reachable from free vertices by eligible alternating paths also have the same parity.
Since $w(e)$ is even (by assumption) and $yz(e)$ is even (by parity) we can conclude that $yz(e) \geq w(e) + 2$ before dual adjustment,
and therefore $yz(e) \geq w(e)$ after dual adjustment.

	\item $u$ but not $v$ is in $\Vin \cup \Vout$ and $e \in M$. This case cannot happen since in this case, 
	$u \in \Vin$ and $e$ must be ineligible, but we know all matched edges are tight.

	\item $u$ but not $v$ is in $\Vin \cup \Vout$ and $e \notin M$. If $u \in \Vin$, then $yz(e)$ increases and domination holds. Otherwise, $u \in \Vout$ and $e$ must be ineligible. In this case, we have $yz(e) \geq w(e) + 1$ before the dual adjustment and 
	$yz(e) \geq w(e)$ afterwards.
\end{enumerate}
\end{proof}

\begin{lemma}\label{lem:offendRCS}
If Property \ref{prop:RCS} is satisfied and the $y$-values of vertices in $F$ have the same parity, 
then 
$\SearchOne(F)$ (under Criterion \ref{crit2})
or
$\EdmondsSearch(F)$ (under Criterion \ref{crit3})
preserves Property~\ref{prop:RCS}.
\end{lemma}

\begin{proof}
The proof is similar to that of the previous lemma, except that we replace the tightness and domination by near tightness and near domination. We point out the differences in the following.
An edge $e$ can be included in a blossom only if it is eligible. An eligible edge must have $yz(e) = w(e)$ or $yz(e) = w(e) - 2$.  Augmentations only make eligible edges  matched. Therefore near tightness is satisfied after the Augmentation step.

When doing the dual adjustment, the following are the cases when $yz(e)$ is modified after the dual adjustment. 
In Case 2 of the previous proof, when $u,v \in \Vout$ but $e$ is ineligible we have $yz(e) \ge w(e)-1$.  By parity this implies that 
$yz(e) \geq w(e)$ before the dual adjustment and $yz(e) \geq w(e) -2$ afterwards. 
Case 3 may happen in this situation. It is possible that $u \in \Vin$ and $e \in M$ is ineligible. 
Then we must have $yz(e) \leq w(e) - 1$ before the dual adjustment and $yz(e) \leq w(e)$ afterwards. 
In Case 4, when $u \in \Vout$, we have $yz(e) \geq w(e) - 1$ before the dual adjustment and $yz(e) \geq w(e) - 2$ afterwards.
\end{proof}

\section{The \Liquidationist{} Algorithm}\label{sect:liquidationist}

The most expedient way to get rid of an inherited blossom is to {\em liquidate} it (our term)
by distributing its $z$-value over its constituents' $y$-values, preserving Property~\ref{prop:CS} (domination). 

\vcm[.3]
\centerline{\framebox{\begin{minipage}{3in}
\Liquidate($B$) :\\
\hcm[.7] $y(u) \leftarrow y(u) + z(B)/2$ \ for all $u\in B$\\
\hcm[.7] $z(B) \leftarrow 0$  \mbox{\ \ (and dissolve $B$)}
\end{minipage}}}
\vcm[.3]

From the perspective of a single edge, liquidation has no effect on $yz(e)$ if $e$ is fully inside $B$ or outside $B$,
but it increases $yz(e)$ by $z(B)/2$ if $e$ straddles $B$.  
From a global perspective, liquidation increases the dual objective $yz(V)$
by $|B|\cdot z(B)/2 - \floor{|B|/2}\cdot z(B) = z(B)/2$.  Since $z(B)$ is generally 
unbounded (as a function of $n$), this apparently destroys the key advantage of scaling algorithms, that
$yz(V)$ is within $O(n)$ of optimum.  It is for this reason that~\cite{G85,GT91} did not pursue liquidation.

The \Liquidationist{} algorithm (see Figure~\ref{fig:liquidationist}) is so named because it liquidates all inherited blossoms.  
Let $w',y',z', M',\Omega'$ be the edge weights, dual variables, matching and blossom set at the end of the $(i-1)$th scale.\footnote{In the first scale, $w',y',z'=0$ and $M',\Omega'=\emptyset$, which satisfies Property~\ref{prop:RCS}.}
Recall that a blossom is {\em large} if it contains at least $\tau$ vertices and {\em small} otherwise.  

The first step is to compute the {\em even} weight function $w$ for the $i$th scale and starting duals $y,z$, as follows.
\begin{align*}
w(e)	& \leftarrow 2(w'(e) + \mbox{the $i$th bit of $\bar{w}(e)$}),\\
y(u)	& \leftarrow 2y'(u)+3,\\
z(B)	& \leftarrow 2z'(B).
\end{align*}  
Lemma~\ref{thm:negative} proves that if $w',y',z'$ satisfy Property~\ref{prop:RCS} w.r.t.~$M'$,	
then $w,y,z$ satisfy Property~\ref{prop:CS} w.r.t.~$M=\emptyset$, except for the Active Blossom property, 
a point that will be moot once we liquidate all blossoms in $\Omega'$.\footnote{At this point 
we continue to use the term ``blossom'' to refer to a $B \in \Omega'$ with $z(B)>0$.
Of course, since $M=\emptyset$, $E_B$ no longer satisfies the structural definition of a blossom, i.e., 
consisting of an odd-length cycle of vertices/subblossoms that alternates between $M$ and $E\backslash M$.}
It will be guaranteed that $\sum_{\operatorname{Large}\, B\in \Omega'} z(B) = O(n)$, 
so liquidating all large blossoms increases $yz(V)$ by a tolerable $O(n)$.
After liquidating large blossoms, but before liquidating small blossoms, we {\em reweight} the graph, setting
\begin{align*}
w(u,v) &\leftarrow w(u,v) - y(u) - y(v)		& \mbox{for each edge $(u,v)$}\\
y(u)	&\leftarrow 0					& \mbox{for each vertex $u$}
\end{align*}
Reweighting is a conceptual trick that simplifies the presentation and some proofs.
A practical implementation would simulate this step without actually modifying the edge weights.

Liquidating small blossoms increases 
$y(u)$ from 0 to $\sum_{\operatorname{Small}\, B\in\Omega',\, u\in B} z(B)/2$, 
which {\em temporarily} destroys the property that $yz(V)$ is within $O(n)$ of optimal.
Let $B'$ be a {\em maximal} former small blossom.  We repeatedly execute $\PQSearch(F)$
from the set $F$ of free vertices in $B'$ with maximum $y$-value $Y$
until one of three events occurs 
(i) $|F|$ decreases, because an augmenting path is discovered, 
(ii) $|F|$ increases because $Y-Y'$ dual adjustments have been performed, 
where $Y'$ is the 2nd largest $y$-value of a free vertex in $B'$,
or (iii) the $y$-values of all vertices in $F$ become zero.
Because $B'$ is small there can be at most $O(|B'|) = O(\tau)$ executions that stop due to (i) and (ii).
We prove that conducting Edmonds' searches in exactly this way has two useful properties.
First, no edge straddling $B'$ ever becomes eligible, so the search is {\em confined to} the subgraph induced by $B'$,		
and second, when the $y$-values of free vertices are zero, $yz(V)$ is restored to be within $O(n)$ of optimal.
Each of these Edmonds' searches can form new weighted blossoms, but because of the first property they
all must be {\em small}.  The second property is essential for the next step: efficiently finding a near-perfect matching.

After inherited blossoms have been dealt with, we switch from satisfying Property~\ref{prop:CS} to Property~\ref{prop:RCS}   
and call $\SearchOne(F)$ $\tau$ times using eligibility Criterion~\ref{crit2}, where $F$ is the set of all free vertices.
We prove that this leaves at most $O(n/\tau)$ free vertices.  
Note that large blossoms can only be introduced during the calls to $\SearchOne$.  Since we only perform
$\tau$ dual adjustments, we can bound the sum of $z$-values of all new large blossoms by $O(n)$.

To end the $i$th scale we artificially match up all free vertices with dummy vertices and zero-weight edges, 
yielding a perfect matching.
Thus, the input graph $G_{i+1}$ to scale $i+1$ is always $G$ supplemented with 
some dummy pendants (degree one vertices) that have accrued over scales 1 through $i$.  Pendants can never appear in a blossom.

After the last scale, we have a perfect matching $M$ in $G_{\ceil{\log((\f{n}{2}+1)N)}}$, 
which includes up to $O(n/\tau)\cdot \log((\f{n}{2}+1)N)$ dummy vertices acquired over all the scales.
We delete all dummy vertices and repeatedly call $\PQSearch(F)$ on the current set of free vertices
until $F=\emptyset$.  Since these calls make many dual adjustments, we switch from Criterion~\ref{crit2} (which is only suitable
for use with $\SearchOne$) to Criterion~\ref{crit3} of eligibility.
Each call to $\PQSearch$ matches at least two vertices so the total time for finalization is 
$O(\Edmonds\cdot (n/\tau)\log(nN))$.  
See Figure~\ref{fig:liquidationist} for a compact summary of the whole algorithm.

\begin{figure}
\centering
\framebox{
\begin{minipage}{6in}
\small
\noindent $\Liquidationist(G, \hat{w})$ 

\begin{itemize}[itemsep=0ex]

	\item $G_0 \leftarrow G, y \leftarrow 0, z \leftarrow 0, w \leftarrow 0$, $\Omega\leftarrow \emptyset$.
	
	\item For scales $i=1,\cdots,\ceil{\log((\f{n}{2}+1)N)}$, execute steps {\bf Initialization}--{\bf Perfection}.

	\item[] {\bf Initialization}
	\begin{enumerate}[topsep=0pt,itemsep=0ex,partopsep=0ex,parsep=0ex]
	\item Set $G_i \leftarrow G_{i-1}$, $y' \leftarrow y$, $z' \leftarrow z$, $w'\leftarrow w$, $M \leftarrow \emptyset$, $\Omega' \leftarrow \Omega$, and $\Omega \leftarrow \emptyset$.
	\end{enumerate}

	\item[] {\bf Scaling}
	\begin{enumerate}[topsep=0pt,itemsep=0ex,partopsep=0ex,parsep=0ex]
		\setcounter{enumi}{1}
		\item \label{item:s2} Set $w(e) \leftarrow 2(w'(e)+\mbox{(the $i^{\operatorname{th}}$ bit of $\bar{w}(e)$)})$
						for each edge $e$,  set $y(u) \leftarrow 2y'(u)+3$ for each vertex $u$, and $z(B') \leftarrow 2z'(B')$ for each $B'\in \Omega'$.
	\end{enumerate}

	\item[] {\bf Large Blossom Liquidation and Reweighting}
	\begin{enumerate}[topsep=0pt,itemsep=0ex,partopsep=0ex,parsep=0ex]\setcounter{enumi}{2}
		\item \label{item:s3} $\Liquidate(B')$ for each large $B'\in \Omega'$.

		\item \label{item:s4} Reweight the graph: 
		\vcm[-.3]
		\begin{align*}
			w(u,v) &\leftarrow w(u,v)-y(u)-y(v) && \mbox{for each edge $(u,v) \in E$} \\
		 	y(u)   &\leftarrow 0             && \mbox{for each vertex $u \in V$}
		\end{align*}

	\end{enumerate}

	\item[] {\bf Small Blossom Liquidation}
	\begin{enumerate}[topsep=0pt,itemsep=0ex,partopsep=0ex,parsep=0ex]\setcounter{enumi}{4}
	\item \label{item:s5} $\Liquidate(B')$ for each small $B'\in \Omega'$.

	\item \label{item:s6} For each {\em maximal} old small blossom $B'$:\\
	\hcm[.5] While $\max\{y(u) \;|\; u\in B' \mbox{ is free}\} > 0$,
			\vcm[-.2]
			\begin{align*} 
			Y &\leftarrow \max\{y(u) \,|\, u\in B' \mbox{ is free}\}\\
			F &\leftarrow \{u\in B' \mbox{ is free} \,|\, y(u) = Y\}\\
			Y' &\leftarrow \max\{0, \max\{y(u) \,|\, u\in B'\backslash F \mbox{ is free}\}\} \hcm[3.9]
			\end{align*}
			\hcm[1.1] Run $\PQSearch(F)$ (Criterion~\ref{crit1}) until an augmenting path is found and \\		
			\hcm[1.1] the matching is augmented or $Y-Y'$ dual adjustments have been performed. \\		

	\end{enumerate}

	\item[] {\bf Free Vertex Reduction}

	\begin{enumerate}[topsep=0pt,itemsep=0ex,partopsep=0ex,parsep=0ex]\setcounter{enumi}{6}
		\item\label{item:s7} Run $\SearchOne(F)$ (Criterion~\ref{crit2}) $\tau$ times, where $F$ is the set of free vertices.
	\end{enumerate}

	\item[] {\bf Perfection}
	\begin{enumerate}[topsep=0pt,itemsep=0ex,partopsep=0ex,parsep=0ex]\setcounter{enumi}{7}
		\item Delete all free dummy vertices.  For each remaining free vertex $u$, create a dummy $\hat{u}$ with $y(\hat{u})=\tau$ 
		and a zero-weight matched edge $(u,\hat{u})\in M$.\\
	\end{enumerate}

	\item {\bf Finalization}\\
	Delete all dummy vertices from $G_{\ceil{\log((\f{n}{2}+1)N)}}$.  
	Repeatedly call $\PQSearch(F)$ (Criterion~\ref{crit3}) on the set $F$ of free vertices until $F=\emptyset$.
\end{itemize}
\end{minipage}
}
\caption{\label{fig:liquidationist}}
\end{figure}

\subsection{Correctness}\label{sec:correct}

We first show that rescaling $w,y,z$ at the beginning of a scale restores Property~\ref{prop:CS} (except for Active Blossoms)
assuming Property~\ref{prop:RCS} held at the end of the previous scale.

\begin{lemma}\label{thm:negative}
Consider an edge $e \in E(G_i)$ at scale $i$.
\begin{itemize}
	\item After Step \ref{item:s2} (Scaling), $w(e)\leq yz(e)$. Moreover, if $e \in M' \cup \bigcup_{B'\in \Omega'} E_{B'}$ then $w(e)\geq yz(e)-6$. 
	(In the first scale, $w(e)\geq yz(e)-6$ for every $e$.)

	\item After Step \ref{item:s4} (Large Blossom Liquidation and Reweighting), $w(e)$ is even for all $e\in E(G_i)$
	and $y(u) = 0$ for all $y\in V(G_i)$.  Furthermore,
\[
w(e) \leq yz(e) = \sum_{\substack{\operatorname{Small}\, B'\in\Omega' \;:\\ e\in E(B')}} z(B').  
\]
\end{itemize}
Therefore, after Large Blossom Liquidation and Reweighting,
$(M,\Omega,y,z,w)$ satisfy Property \ref{prop:CS}, excluding Active Blossoms.
\end{lemma}

\begin{proof}
At the end of the previous scale, by Property~\ref{prop:RCS}(near domination), $y'z'(e)\geq w'(e)-2$.  
After the Scaling step, 
\[
yz(e)=2y'z'(e)+6   \geq   2w'(e)+2   \geq  w(e).
\]
If $e \in M' \cup \bigcup_{B'\in \Omega'} E_{B'}$ was an old matching or blossom edge then
\[
yz(e)=2y'z'(e)+6\leq 2w'(e)+6\leq w(e)+6.
\]
In the first scale, $yz(e)=6$ and $w(e) \in\{0,2\}$.
Step \ref{item:s3} will increase some $yz$-values and $w(e)\leq yz(e)$ will be maintained. 
After Step \ref{item:s4} (reweighting), 
$w(u,v)$ will be reduced by $y(u)+y(v)$, so
\[
w(u,v) \leq \sum_{\substack{\operatorname{Small}\, B'\in\Omega':\\(u,v)\in E(B')}} z(B'). 
\]
From Property~\ref{prop:RCS}(1) (granularity) in the previous scale, after Step \ref{item:s2}
all $y$-values are odd and $z$-values are multiples of 4.
Therefore $y$-values remain odd after Step \ref{item:s3}.  
Since $w(u,v)$ is even initially, it remains even after subtracting off odd $y(u),y(v)$ in Step \ref{item:s4}.
\end{proof}

\begin{lemma}\label{lem:smallblossomelimination}
After Step \ref{item:s5} in Small Blossom Liquidation, we have $w(u,v)\leq 2 \cdot \min\{y(u),y(v)\}$ for all edges $(u,v)$,		
hence, $w(u,v)\leq y(u)+y(v)$. Furthermore, Property \ref{prop:CS} holds after Small Blossom Liquidation.
\end{lemma}
\begin{proof}
Fix any edge $(u,v)$.  According to Lemma \ref{thm:negative},
after Step \ref{item:s5} we have
\[
y(u)=\sum_{\substack{\operatorname{Small}\, B'\in\Omega':\\u\in B'}}z(B')/2 
\;\geq\; \sum_{\substack{\operatorname{Small}\, B'\in\Omega':\\(u,v)\in E(B')}}z(B')/2 
\;\geq\; w(u,v)/2.
\]
Therefore, $w(u,v)\leq 2y(u)$ and by symmetry, $w(u,v)\leq 2y(v)$. After Step \ref{item:s5}, $z=0$ and $\Omega=\emptyset$,
so Property~\ref{prop:CS} (including Active Blossoms) holds. 
In Step \ref{item:s6} of Small Blossom Liquidation, $\PQSearch$ is always searching from the free vertices with the 
same $y$-values and the edge weights are even. Therefore, by Lemma~\ref{lem:offendCS}, 
Property \ref{prop:CS} holds afterwards.
\end{proof}

\begin{lemma}\label{lem:liquidationist-correct}
The \Liquidationist{} algorithm returns a maximum weight perfect matching of $G$.		
\end{lemma}
\begin{proof}
First we claim that at the end of each scale $i$, $M$ is a perfect matching in $G_i$ and Property \ref{prop:RCS} is satisfied. By Lemma \ref{lem:smallblossomelimination}, Property \ref{prop:CS} is satisfied after the Small Blossom Liquidation step. The calls to $\SearchOne$ in the Free Vertex Reduction step always search from free vertices with the same $y$-values. Therefore, by Lemma~\ref{lem:offendRCS}, Property \ref{prop:RCS} holds afterwards. The perfection step adds/deletes dummy free vertices and edges to make the matching $M$ perfect. The newly added edges have $w(e) = yz(e)$, and so Property \ref{prop:RCS} is maintained at the end of scale $i$.

Therefore, Property \ref{prop:RCS} is satisfied at the end of the last scale $\lceil\log((\fr{n}{2}+1)N)\rceil$.  
Consider the shrunken blossom edges at this point in the algorithm.  
Each edge $e$ was made a blossom edge when it was eligible according to Criterion~\ref{crit1} (in Step~\ref{item:s6}) or Criterion~\ref{crit2} (in Step~\ref{item:s7})
and may have participated in augmenting paths while its blossom was still shrunken.  Thus, all we can claim 
is that $yz(e) - w(e) \in \{0,-2\}$.  In the calls to $\PQSearch$ in the Finalization step we switch to eligibility 
Criterion~\ref{crit3} in order to ensure that edges inside shrunken blossoms {\em remain} eligible whenever 
the blossoms are dissolved in the course of the search.\footnote{Alternatively, we could use Criterion~\ref{crit2} but allow all formerly shrunken blossom edges to be automatically eligible.}
By Lemma~\ref{lem:offendRCS}, each call to $\PQSearch$ maintains Property \ref{prop:RCS} while making the matching perfect. 
After Finalization, $w(M)\geq w(M^*)-n$.  
Note that in the last scale ${w}(e)=2\bar{w}(e)$ for each edge $e$, so $\bar{w}(M)\geq \bar{w}(M^*)-n/2$. 
By definition of $\bar{w}$, $\bar{w}(M)$ is a multiple of $\fr{n}{2}+1$, so $M$ maximizes $\bar{w}(M)$ 
and hence $\hat{w}(M)$ as well.
\end{proof}

\subsection{Running time}\label{sec:time}

Next, we analyze the running time. 

\begin{lemma}\label{thm:z-small}
In Step $\ref{item:s6}$, we only need to consider the edges within small blossoms of the previous scale.
The total time needed for Step $\ref{item:s6}$ in one scale is 
$O((m+n\log n)\tau)$ (using~\cite{Gabow16})
or $O(m\sqrt{\log\log n}\cdot \tau)$ w.h.p. (see Section~\ref{sect:ImplementingEdmonds}).
\end{lemma}

\begin{proof}
We first analyze the behavior of Step 6 assuming we only consider edges with both endpoints in the same maximal small blossom,
i.e., straddling edges are ignored.  Then we argue that straddling edges can never become eligible in Step 6, so a correct implementation
may ignore them.

Let $Y$ denote the \emph{current} maximum $y$-value of a free vertex in a maximal small blossom $B$ processed in Step 6.
We prove by induction that the $y$-values of all vertices in $\Vin \cup \Vout$ are at least $Y$.
The proof is by induction.
After Initialization, since $M=\emptyset$, we have $\Vin \cup \Vout = F$. 
Suppose that it is true before a dual adjustment in $\PQSearch(F)$. 
After the dual adjustment, the maximum $y$-value of a free vertex is now $Y-1$. 
Vertices can have their $y$-values decreased by at most one which may cause new edges straddling $\Vin \cup \Vout$ to become eligible.
Suppose that $(u,v) \in E(B)$ becomes eligible after the dual adjustment, adding $v\in B$ to the set $\Vin \cup \Vout$.
The eligibility criterion is {\em tightness} (Criterion~\ref{crit1}), so
we must have $w(u,v) = y(u) + y(v) \geq (Y-1) + y(v)$. 
On the other hand, by Lemma \ref{lem:smallblossomelimination} and since 
$y(v)$ has not been changed since Step \ref{item:s5}, we have $w(u,v) \leq 2y(v)$. 
Therefore, $y(v) \geq Y-1$.

Now consider an edge $(u,v)$ with $u$ and $v$ in different maximal small blossoms.  Just before we liquidate small blossoms,
$w(u,v) \le yz(u,v) = y(u) + y(v) = 0$; after liquidation we have $yz(u,v) = Y_u + Y_v$ where 
$Y_x = \sum_{\operatorname{Small}\, B\in \Omega' : x\in B} z(B)/2$.  As argued above, when we process $u$'s (resp., $v$'s)
maximal small blossom, $u$ (resp., $v$) will participate in at most $Y_u$ (resp., $Y_v$) 
dual adjustments before the free vertices' $y$-values reach zero.  Thus, $(u,v)$ will never become eligible during any search in Step 6.

Thus, we only consider the edge set $E(B')$ when processing $B'$ in Step 6. 
Sorting the $y$-values takes $O(n\log n)$ time. 	
Before $Y$ reaches 0, each call to $\PQSearch(F)$ takes $O(m(B') + n(B')\log n(B'))$ time (using \cite{Gabow16})
or $O(m(B')\sqrt{\log\log n(B')})$ time w.h.p. (Section~\ref{sect:ImplementingEdmonds})
and either matches at least two more 
vertices in $B'$ or enlarges the set $F$ of free vertices with maximum $y$-value in $B'$.  
Thus there can be at most $O(n(B')) = O(\tau)$ calls to $\PQSearch$ on $B'$.  
Summed over all maximal small $B'\in \Omega'$,
the total time for Step~\ref{item:s6} is $O((m+n\log n)\tau)$ or $O(m\sqrt{\log\log n}\cdot \tau)$ w.h.p.
\end{proof}

\begin{lemma}\label{thm:z-large}
The sum of $z$-values of large blossoms at the end of a scale is at most $2n$.
\end{lemma}
\begin{proof}
By Lemma~\ref{thm:z-small}, 
Small Blossom Liquidation only operates on subgraphs of at most $\tau$ vertices and therefore cannot create any large blossoms. Every dual adjustment  performed in the Free Vertex Reduction step 
increases the $z$-values of at most $n / \tau$ large root blossoms, each by exactly 2.
(The dummy vertices introduced in the Perfection step of scales $1$ through $i-1$ are pendants and cannot be in any blossom.
Thus, the `$n$' here refers to the number of original vertices, not $|V(G_i)|$.)
There are at most $\tau$ dual adjustments in Free Vertex Reduction, which implies the lemma.
\end{proof}

\begin{lemma}\label{lem:comparematching}
Let $M'$ be the perfect matching obtained in the previous scale. Let $M''$ be any (not necessarily perfect) matching. 
After Large Blossom Liquidation we have $w(M'') \leq w(M') + 8n$.
\end{lemma}
\begin{proof}
Consider the perfect matching $M'$ obtained in the previous scale, whose blossom set $\Omega'$ is partitioned into small and large blossoms.  (For the first scale, $M'$ is any perfect matching and $\Omega'=\emptyset$.)
Define $K$ to be the increase in the dual objective due to Large Blossom Liquidation,
\[
K \;=\;  \sum_{\operatorname{Large}\, B'\in \Omega'} z(B')/2 \;=\; \sum_{\operatorname{Large}\, B'\in \Omega'} z'(B').
\]
By Lemma~\ref{thm:z-large}, $K \leq 2n$.  Let $y_i,z_i$ denote the duals after Step $i$ of \Liquidationist. Let $w_0$ be the weight function before Step 4 (reweighting) 
and $w$ be the weight afterwards. 
We have:
\begin{align*}
w_0(M') & \geq  - 6|M'| + \sum_{e \in M'} y_2z_2 (e) && \mbox{Lemma \ref{thm:negative}}\\
&=  - 6 |M'| - K + \sum_{e \in M'} y_3z_3 (e) && \mbox{see above}\\
w(M') &\geq - 6|M'| - 2n + \sum_{e \in M'} y_4z_4(e) \\
& \geq  -6|M'|- 2n+\sum_{\operatorname{Small}\, B'\in \Omega'}z_4(B')\cdot \floor{|B'|/2} && \mbox{Since $y_4=0$}\\
&\geq -8n + \sum_{\operatorname{Small}\, B'\in \Omega'}z_4(B') \cdot \lfloor|B'|/2\rfloor && \mbox{(\#dummy vertices) $\leq n$ } \\
&\geq -8n + \sum_{e \in M''}\sum_{\substack{\operatorname{Small}\, B'\in\Omega': \\ e \in E(B') }}z_4(B') && \mbox{$|M'' \cap E(B')| \leq \lfloor |B'| / 2\rfloor$} \\
&\geq -8n + \sum_{e \in M''}w(e) \;=\; -8n + w(M'')  && \mbox{by Lemma \ref{thm:negative}}
\end{align*}
Observe that this Lemma would not be true as stated without the Reweighting step, which allows us to directly
compare the weight of perfect and imperfect matchings.
\end{proof}

The next lemma is stated in a more general fashion than is necessary so that we can apply it again later, in Section~\ref{sect:hybrid}.
In the \Liquidationist{} algorithm, after Step~\ref{item:s6} all $y$-values of free vertices are zero, so the
sum $\sum_{u \notin V(M)}y_6(u)$ seen below vanishes.

\begin{lemma}\label{thm:number-free} 
Let $y_6, z_6$ be the duals after Step~\ref{item:s6}, just before the Free Vertex Reduction step. 
Let $M$ be the matching after Free Vertex Reduction and $f$ be the number of free vertices with respect to $M$. Suppose that there exists some perfect matching $M'$ such that $w(M) \leq w(M') + 8n - \sum_{u \notin V(M)}y_6(u)$. Then, $f \leq 10n/\tau$.
\end{lemma}
\begin{proof}
Let  $y_7,z_7,\Omega$ denote the duals and blossom set {\em after} Free Vertex Reduction. 
By Property~\ref{prop:RCS},
\begin{align*}
w(M')&\leq \sum_{e \in M'} (y_7z_7(e) + 2) 				&& \mbox{near domination}\\
&= \sum_{u\in V} y_7(u)+\sum_{e\in M'}\sum_{\substack{B\in\Omega:\\e\in E(B)}} z_7(B)+2|M'| &&  \\
&\leq \sum_{u\in V} y_7(u)+\sum_{B\in\Omega}z_7(B)\cdot \floor{|B|/2}+2n 			&& \mbox{(\#dummy vertices) $\leq n$ }\\
&\leq \left(\sum_{u \in V(M)} y_7(u) + \sum_{B\in\Omega}z_7(B)\cdot \floor{|B|/2} \right) + \sum_{u \notin V(M)} y_7(u) +  2n\\
&= \sum_{e \in M}y_7z_7(e) + \sum_{u \notin V(M)} y_7(u) +  2n\\
&\leq w(M) + \sum_{u \notin V(M)} y_7(u) +  2n && \mbox{near tightness}\\
&= w(M) + \sum_{u \notin V(M)} y_6(u) - f \tau +  2n   && y_7(u) = y_6(u) - \tau\\
&\le w(M') + 10n - f\tau					&& \mbox{by assumption of $M'$}	
\end{align*}
Therefore, $f\tau \leq 10n$, and $f\leq 10n / \tau$.
\end{proof}

\begin{theorem}
The \Liquidationist{} algorithm runs in $O((m+n\log n)\sqrt{n}\log(nN))$ time,
or $O(m\sqrt{n\log\log n}\log(nN))$ time w.h.p.
\end{theorem}
\begin{proof}
Initialization, Scaling, and Large Blossom Liquidation take $O(n)$ time.
By Lemma~\ref{thm:z-small}, the time needed for Small Blossom Liquidation is 
$O(\Edmonds\cdot \tau)$, where $\Edmonds$ is the cost of one Edmonds' search.
Each iteration of $\SearchOne$ takes $O(m)$ time, so the time needed for Free Vertex Reduction is $O(m \tau)$.
By Lemmas~\ref{lem:comparematching} and~\ref{thm:number-free}, 
at most $10(n/\tau)\ceil{\log((\fr{n}{2}+1)N)}$ free vertices emerge after deleting dummy vertices. 
Since we have rescaled the weights many times, we cannot bound the weighted length of augmenting paths by 
$O(m)$.  The cost for rematching vertices in the Finalization step is $O(\Edmonds \cdot (n/\tau)\log(nN))$.
The total time is therefore
$O((m\tau + \Edmonds\cdot (\tau + n/\tau))\log(nN))$, which is minimized when $\tau=\sqrt{n}$.
Depending on the implementation of $\PQSearch$ this is
$O((m+n\log n)\sqrt{n}\log(nN))$ or $O(m\sqrt{n\log\log n}\log(nN))$ w.h.p.
\end{proof}

\section{The {\sc Hybrid} Algorithm}\label{sect:hybrid}

In this section, we describe an \MWPM{} algorithm called \Hybrid{} that runs in 
$O( m\sqrt{n} \log (nN) )$ time even on sparse graphs. 
In the \Liquidationist{} algorithm, the Small Blossom Liquidation and the Free Vertex Reduction steps 
contribute $O(\Edmonds\cdot \tau)$ and $O(m \tau)$ to the running time. If we could do these steps faster, then it would be possible for us to choose a slightly larger $\tau$, thereby reducing the number
of vertices that emerge free in the Finalization step.
The time needed to rematch these vertices is $O(\Edmonds\cdot (n / \tau)\log (nN))$, 
which is at most $O(m\sqrt{n}\log(nN))$ for, say, $\tau = \sqrt{n}\log n$.

The pseudocode for $\Hybrid$ is given in Figure~\ref{fig:hybrid}.  It differs from the \Liquidationist{} algorithm in two respects.  Rather than do Small Blossom Liquidation, it uses Gabow's method on each maximal small blossom $B'\in \Omega'$
in order to dissolve $B'$ and all its sub-blossoms.  (Lemma~\ref{lem:gabow} lists the salient properties of Gabow's algorithm; it is proved in Section \ref{sec:gabow}.)
The Free Vertex Reduction step is now done in two stages since we cannot afford to call $\SearchOne$ $\tau=\omega(\sqrt{n})$ times.  The first $\sqrt{n}$ dual adjustments are performed by $\SearchOne$ with eligibility Criterion~\ref{crit2} and the remaining $\tau-\sqrt{n}$ dual adjustments are performed in calls to $\BucketSearch$ with eligibility Criterion~\ref{crit3}.\footnote{We switch to Criterion~\ref{crit3} 
to ensure that formerly shrunken blossom edges remain eligible when the blossom is dissolved in the course of a search.  See the discussion in the proof of Lemma~\ref{lem:liquidationist-correct}.}

\begin{lemma}\label{lem:gabow}
Fix a $B\in \Omega'$.  Suppose that Property~\ref{prop:CS} holds, that all free vertices in $B$ have the same parity,
and that $yz(e) \le w(e) + 6$ for all $e\in E_B$.
After calling Gabow's algorithm on $B$ the following hold.
\begin{itemize}[itemsep=0ex]
\item All the old blossoms $B' \subseteq B$ are dissolved. 
\item Property \ref{prop:CS} holds and the $y$-values of free vertices in $B$ have the same parity.
\item $yz(V)$ does not increase.
\end{itemize}
Futhermore, Gabow's algorithm runs in $O(m(B) (n(B))^{3/4} )$ time.
\end{lemma}

\begin{figure}
\centering
\framebox{
\begin{minipage}{6.2in}
\noindent $\Hybrid(G, \hat{w})$ 

\begin{itemize}[topsep=1ex,itemsep=1ex,partopsep=0ex,parsep=0ex]

	\item $G_0 \leftarrow G, y \leftarrow 0, z \leftarrow 0, w\leftarrow 0, \Omega\leftarrow \emptyset$.
	
	\item For scales $i=1,\cdots,\lceil\log((\fr{n}{2}+1)N)\rceil$, execute steps {\bf Initialization} through {\bf Perfection}.

	\item[] {\bf Initialization}, {\bf Scaling}, and {\bf Large Blossom Liquidation} are performed exactly as in \Liquidationist.
	(There is no need to do Reweighting after Large Blossom Liquidation.)

	\item[] {\bf Small Blossom Dissolution}

	\begin{enumerate}\setcounter{enumi}{0}
	\item \label{item:gabowsalgorithm} Run Gabow's algorithm on each maximal small blossom $B'\in\Omega'$.
	\end{enumerate}
	
	\item[] {\bf Free Vertex Reduction}

	\begin{enumerate}\setcounter{enumi}{1}
		\item[] Let $F$ always denote the current set of free vertices and $\delta$ the number of dual adjustments 
		performed so far in Steps~\ref{item:searchone} and \ref{item:bucketsearch}.
		\item\label{item:searchone} Run $\SearchOne(F)$ (Criterion~\ref{crit2}) $\sqrt{n}$ times.
		\item\label{item:bucketsearch} While $\delta < \tau$ and $M$ is not perfect, 
		call $\BucketSearch(F)$ (Criterion~\ref{crit3}), terminating when an augmenting path is found and the matching is augmented or when $\delta = \tau$.		
	\end{enumerate}
	\item[] {\bf Perfection} is performed as in \Liquidationist.

	\item {\bf Finalization} is performed as in \Liquidationist.
\end{itemize}
\end{minipage}
}
\caption{\label{fig:hybrid}}
\end{figure}


\subsection{Correctness and Running Time}
We first argue scale $i$ functions correctly.  
Assuming Property~\ref{prop:RCS} holds at the end of scale $i-1$,
Property \ref{prop:CS} (except Active Blossoms) holds after Initialization at scale $i$.  
Note that Lemma \ref{thm:z-large} was not sensitive to the value of $\tau$, so it holds for \Hybrid{} as well as \Liquidationist. 
We can conclude that Large Blossom Liquidation increases the dual objective by 
$\sum_{\operatorname{Large}\, B'\in\Omega'} z'(B') \le 2n$.
By Lemma~\ref{lem:gabow}, the Small Blossom Dissolution step dissolves all remaining old blossoms and restores Property \ref{prop:CS}. 
By Lemma~\ref{lem:offendRCS}, the Free Vertex Reduction step maintains Property \ref{prop:RCS}. 	
The rest of the argument is the same as in Section \ref{sec:correct}.

In order to bound the running time we need to prove that the Free Vertex Reduction step runs in $O(m\sqrt{n})$ time, independent of $\tau$, and that afterwards there are at most $O(n/\tau)$ free vertices.

We now prove a lemma similar to Lemma \ref{lem:comparematching} that allows us to apply Lemma \ref{thm:number-free}.
\begin{lemma}\label{lem:comparematching2}
Let $M'$ be the perfect matching obtained in the previous scale and $M''$ be any matching, not necessarily perfect. 
We have $w(M'') \leq w(M') + 8n - \sum_{u \notin V(M'')} y(u)$ after the Small Blossom Dissolution step of \Hybrid.
\end{lemma}
\begin{proof}

Let $y_0,z_0$ denote the duals immediately before Small Blossom Dissolution and
$y,z,\Omega$ denote the duals and blossom set after Small Blossom Dissolution. Similar to the proof of Lemma \ref{lem:comparematching}, we have, for $K = \sum_{\operatorname{Large}\, B'\in \Omega'} z'(B')$,
\begin{align*}
w(M') &\geq - 6|M'| - K + \sum_{e \in M'} y_0z_0(e) 			&& \mbox{Lemma~\ref{thm:negative}}\\
&\geq  -8n+y_0z_0(V) 									&& K \le 2n\\			
&\geq -8n + yz(V) && \mbox{By Lemma \ref{lem:gabow}}\\ 
&= -8n + \sum_{u \in V}y(u) + \sum_{B' \in \Omega}z(B')\cdot \lfloor|B'| / 2 \rfloor \\
&\geq -8n + \sum_{u \notin V(M'')}y(u) + \left( \sum_{u \in V(M'')}y(u) + \sum_{B' \in \Omega}z(B')\cdot \lfloor|B'| / 2 \rfloor \right) \\
&\geq -8n + \sum_{u \notin V(M'')}y(u) + \sum_{e \in M''} yz(e) \\
&\geq -8n + \sum_{u \notin V(M'')}y(u) + w(M'') && \mbox{Property \ref{prop:CS} (domination)}
\end{align*}

\end{proof}

Therefore, because Lemma \ref{lem:comparematching2} holds for any matching $M''$, we can apply Lemma \ref{thm:number-free} to show the number of free vertices after Free Vertex Reduction is bounded by $O(n/\tau)$.


\begin{theorem}
\Hybrid{} computes an \MWPM{} in time
\[
O\left(\left[m\sqrt{n} + m \tau^{3/4} + \Edmonds\cdot (n/\tau)\right] \log(nN)\right).
\]
\end{theorem}
\begin{proof}
Initialization, Scaling, and Large Blossom Liquidation still take $O(n)$ time.
By Lemma \ref{lem:gabow}, the Small Blossom Dissolution step takes $O(m(B) (n(B))^{3/4})$ time for each maximal small blossom $B\in\Omega'$, for a total of $O(m \tau^{3/4})$. 
We now turn to the Free Vertex Reduction step. 
After $\sqrt{n}$ iterations of $\SearchOne(F)$, we have performed $\ceil{\sqrt{n}}$ units of dual adjustment from all the remaining free vertices. By Lemma \ref{lem:comparematching2} and Lemma \ref{thm:number-free}, there are at most $10n/\ceil{\sqrt{n}} = O(\sqrt{n})$ free vertices.  
Throughout the Free Vertex Reduction step, the difference $w(M)-yz(V)$ is $O(n)$, since $w(M)-w(M')$ is $O(n)$ by Lemma~\ref{lem:comparematching} for the perfect matching $M'$ of the previous scale, and $w(M')-yz(V)$ is $O(n)$ by domination. 		
Since each dual adjustment reduces $yz(V)$ by at least 1, 
we can implement $\BucketSearch$ with an array of $O(n)$ buckets for the priority queue.
See Section~\ref{sect:ImplementingEdmonds} for details.
A call to $\BucketSearch(F)$ that finds $p\ge 0$ augmenting paths takes $O(m(p+1))$ time.
Only the last call to \BucketSearch{} may fail to find at least one augmenting path, 
so the total time for all such calls is $O(m\sqrt{n})$.

By Lemma~\ref{thm:number-free} again, after Free Vertex Reduction, there can be at most $10n/\tau$ free vertices. Therefore, in the Finalization step, at most $(10n/ \tau)\ceil{\log((\fr{n}{2}+1)N)}$ free vertices emerge after deleting
dummy vertices.  It takes $O(\Edmonds\cdot (n / \tau )\log(nN))$ time to rematch them with Edmonds' search.
\end{proof}

Here we can afford to use any reasonably fast $O(m\log n)$ implementation of $\PQSearch$, such 
as~\cite{GalilMG86,GGS89,Gabow16} or the one presented in Section~\ref{sect:ImplementingEdmonds}.  
Setting $\tau \in [\sqrt{n} \log n, n^{2/3}]$, we get a running time of $O(m\sqrt{n} \log (nN))$
with any $O(m\log n)$ implementation of $\PQSearch$.

\subsection{Gabow's Algorithm}\label{sec:gabow}

The input is a maximal old small blossom $B\in\Omega'$ containing no matched edges, where $yz(e) \geq w(e)$ for all $e \in B$ and $yz(e) \leq w(e) + 6$ for all $e \in E_{B}$. Let $T$ denote the old blossom subtree rooted at $B$.  The goal is to dissolve all the old blossoms in $T$ and satisfy Property \ref{prop:CS} 
\emph{without increasing the dual objective value $yz(V)$}. Gabow's algorithm achieves this in $O({m}(B) ({n}(B))^{3/4})$ time. This is formally stated in Lemma \ref{lem:gabow}.

Gabow's algorithm decomposes $T$ into major paths.  Recall that a child $B_1$ of $B_2$ is a {\em major child} if $|B_1| > |B_2|/2$.
A node $R$ is a {\em major path root} if $R$ is not a major child, so $B$ is a major path root.
The {\it major path} $P(R)$ rooted at $R$ is obtained by starting at $R$ and moving to the major child of the current node, so long as it exists. 

Gabow's algorithm is to traverse each node $R$ in $T$ in {\em postorder}, and if $R$ is a major path root, to call 
$\Path(R)$.  The outcome of $\Path(R)$ is that all remaining old sub-blossoms of $R$ are dissolved, including $R$.
Define the {\em rank} of $R$ to be $\floor{\log {n}(R)}$. 
Suppose that $\Path(R)$ takes $O({m}(R) ({n}(R))^{3/4})$ time. 
If blossoms $R$ and $R'$ correspond to major path roots with the same rank, then $R\cap R' = \emptyset$.
In particular, each edge has its endpoints in at most one major path root of each rank.
Thus, summing over all ranks, the total time to dissolve $B$ and its sub-blossoms is therefore 
\[
O\paren{\sum_{r = 1}^{\lfloor \log {n}(B) \rfloor} {m}(B)\cdot (2^{r+1})^{3/4}} = O\paren{({m}(B)({n}(B))^{3/4}}.
\]
Thus, our focus will be on the analysis of $\Path(R)$.  In this algorithm inherited blossoms from $\Omega'$
coexist with new blossoms in $\Omega$.  We enforce a variant of Property~\ref{prop:CS} that additionally 
governs how old and new blossoms interact.

\begin{property}\label{prop:blossom}
Property~\ref{prop:CS}(1,3,4) holds 
and (2) (Active Blossoms) is changed as follows.
Let $\Omega'$ denote the set of as-yet undissolved blossoms from the previous scale and 
$\Omega,M$ be the blossom set and matching from the current scale. 
\begin{enumerate}[itemsep=0ex]
\item[2a.] $\Omega' \cup \Omega$ is a laminar (hierarchically nested) set.
\item[2b.] There do not exist $B\in \Omega, B'\in \Omega'$ with $B' \subseteq B$.
\item[2c.] No $e\in M$ has exactly one endpoint in some $B' \in \Omega'$.
\item[2d.] If $B\in \Omega$ and $z(B)>0$ then $|E_B \cap M| = \floor{|B|/2}$.
An $\Omega$-blossom is called a {\em root blossom} if it is not contained in any other $\Omega$-blossom.  
All root blossoms have positive $z$-values.
\end{enumerate}
\end{property}

\begin{figure}
\centering
\framebox{
\begin{minipage}{6.3in}
$\Path(R)$: $R$ is a major path root. 
\begin{enumerate}
\item[] Let $F$ be the set of free vertices that are still in undissolved blossoms of $P(R)$.		

\item While $P(R)$ contains undissolved blossoms and $|F| \ge 2$,

	\begin{itemize}
	\item Sort the undissolved atomic shells in non-increasing order by the number of free vertices, excluding those with less than 2 free vertices. 		
	Let $S_1, S_2, \ldots, S_k$ be the resulting list.

	\item For $i \leftarrow 1 \ldots k$, call $\ShellSearch(S_i)$ (Criterion~\ref{crit1}).
	\end{itemize}

\item If $P(R)$ contains undissolved blossoms (implying $|F|=1$)
	\begin{itemize}
	\item Let $\omega$ be the free vertex in $R$.  Let $B_1 \subset B_2 \subset \cdots \subset B_\ell$ be the undissolved blossoms in $P(R)$ and $T = \sum_{i} z(B_i)/2$.
	\item For $i = 1,2,\ldots, \ell$, $\Liquidate(B_i)$
	\item Call $\PQSearch(\{\omega\})$ (Criterion~\ref{crit1}), halting after $T$ dual adjustments.
	\end{itemize}
\end{enumerate}
\end{minipage}}

\end{figure}

\begin{figure}
\centering
\framebox{
\begin{minipage}{6in}
	$\ShellSearch(C,D)$ 
	\begin{itemize}[itemsep=0ex]
	\item[]	Let $C^{*} \supseteq C$ be the smallest undissolved blossom containing $C$.\\
	Let $D^{*} \subseteq D$ be the largest undissolved blossom contained in $D$, or $\emptyset$ if none exists.\\
	Let $F^*$ be the set of free vertices in $G(C^*,D^*)$.
	
	\item[] 	Repeat Augmentation, Blossom Shrinking, Dual Adjustment, and Blossom Dissolution steps
	\ul{until a halting condition occurs} (enumerated below).  

		\item {\em Augmentation:}\\
		Augment $M$ to contain an \MCM{} in the eligible subgraph of $G(C^{*},D^{*})$ and update $F^*$.	
		(This step may find zero augmenting paths and not change $M$.)
		
		\item {\em Blossom Shrinking}:\\ 
		Find and shrink blossoms reachable from $F^*$, exactly as in Edmonds' algorithm.

		\item {\em Dual Adjustment}:\\
		Peform dual adjustments (from $F^*$) as in Edmonds' algorithm, and perform a unit translation on $C^*$ and $D^*$ as follows:
   		\begin{align*}
			z(C^{*}) &\leftarrow z(C^{*}) - 2  &\\
			z(D^{*}) &\leftarrow z(D^{*}) - 2  & \mbox{if $D^*\neq \emptyset$}\\
			y(u) &\leftarrow y(u) + 2 & \mbox{for all $u \in D^{*}$} \\
			y(u) &\leftarrow y(u) + 1 & \mbox{for all $u \in C^{*} \setminus D^{*}$}
		\end{align*}
	
		\item {\em Blossom Dissolution}:\\
		Dissolve root blossoms in $\Omega$ with zero $z$-values as long as they exist.  In addition,
		\begin{itemize}[itemsep=0ex]
			\item[] If $z(C^{*}) = 0$, \ set $\Omega' \leftarrow \Omega' \backslash \{C^*\}$ \ and update $C^*$.
			\item[] If $z(D^{*}) = 0$, \ set $\Omega' \leftarrow \Omega' \backslash \{D^*\}$ \ and update $D^*$.
			\item[] Update $F^*$ to be the set of free vertices in $G(C^*,D^*)$.
		\end{itemize}
	\end{itemize}
\ \\
\ \\
	\underline{Halting Conditions:}
	\begin{enumerate}[topsep=2pt, itemsep=-0.5ex]
		\item The Augmentation step discovers at least one augmenting path.
		\item $G(C^{*}, D^{*})$ absorbs vertices already searched in the same iteration of $\Path$.
		\item $C^{*}$ was the outermost undissolved blossom and dissolves during Blossom Dissolution.
	\end{enumerate}
\end{minipage}}

\caption{$\ShellSearch(C,D)$}
\label{fig:shellsearch}		
\end{figure}

\subsubsection{The procedure $\Path(R)$}

Because $\Path$ is called on the sub-blossoms of $B$ in {\em postorder}, upon calling $\Path(R)$ the only undissolved blossoms in $R$ are those in $P(R)$.
Let $C,D\in P(R) \cup \{\emptyset\}$ with $C\supset D$.  The subgraph induced by $C\backslash D$ is called a {\em shell}, denoted $G(C,D)$. 
Since all blossoms have an odd number of 
vertices, $G(C,D)$ is an {\em even} size shell if $D\neq \emptyset$ and an {\em odd} size shell if $D=\emptyset$.
Moreover, the number of free vertices in an even (odd) shell is always even (odd).
It is an {\em undissolved shell} if both $C$ and $D$ are undissolved, or $C$ is undissolved and $D=\emptyset$.
We call an undissolved shell {\em atomic} if there is no undissolved blossom $C'\in \Omega'$ with $D \subset C' \subset C$.

The procedure $\Path(R)$ has two stages.  The first consists of {\em iterations}. Each iteration begins by surveying the undissolved blossoms in $P(R)$, say they are 
$B_k \supset B_{k-1} \supset \cdots \supset B_1$.  Let the corresponding atomic shells be $S_i = G(B_{i}, B_{i-1})$, where $B_0 \bydef \emptyset$.
We sort the set of atomic shells $\{S_i\}$ in non-increasing order by their number of free vertices and call 
$\ShellSearch(S_i)$ in this order, 
but refrain from making the call unless $S_i$ contains at least two free vertices.

The procedure $\ShellSearch(C,D)$ (see Figure~\ref{fig:shellsearch}) is simply an instantiation of \EdmondsSearch{} with the following features and differences.	
\begin{enumerate}
\item There is a {\em current atomic shell} $G(C^*,D^*)$, which is initially $G(C,D)$, and the Augmentation, Blossom Shrinking, and Dual Adjustment steps		
only search from the set of free vertices in the current atomic shell.  By definition $C^*$ is the smallest undissolved blossom containing $C$ and $D^*$ the largest undissolved blossom contained in $D$, or $\emptyset$ if no such blossom exists.	
\item An edge is eligible if it is tight (Criterion~\ref{crit1}) and in the current atomic shell.  Tight edges that straddle the shell are specifically excluded.
\item Each unit of dual adjustment is accompanied by a unit {\em translation} of $C^*$ and $D^*$, if $D^*\neq \emptyset$.  This may cause either/both of $C^*$ and $D^*$
to dissolve if their $z$-values become zero, which then causes the current atomic shell to be updated.\footnote{To {\em translate} a blossom $B$ by one unit means to decrement $z(B)$ by 2 and increment $y(u)$ by 1 for each $u\in B$.}
See Figure~\ref{fig:Shell}.
\item Like $\EdmondsSearch$, $\ShellSearch$ halts after the first Augmentation step that discovers an augmenting path.
However, it halts in two other situations as well.  If $C^*$ is the outermost undissolved blossom in $P(R)$ and $C^*$ dissolves, $\ShellSearch$ halts immediately.
If the current shell $G(C^*,D^*)$ ever intersects a shell searched in the same iteration of $\Path(R)$, $\ShellSearch$ halts immediately.
Therefore, at the end of an iteration of $\Path(R)$, every undissolved atomic shell contains at least
two vertices that were matched (via an augmenting path) in the iteration.
\end{enumerate}

\begin{figure}
\centering
\begin{tabular}{cc}
\scalebox{.4}{\includegraphics{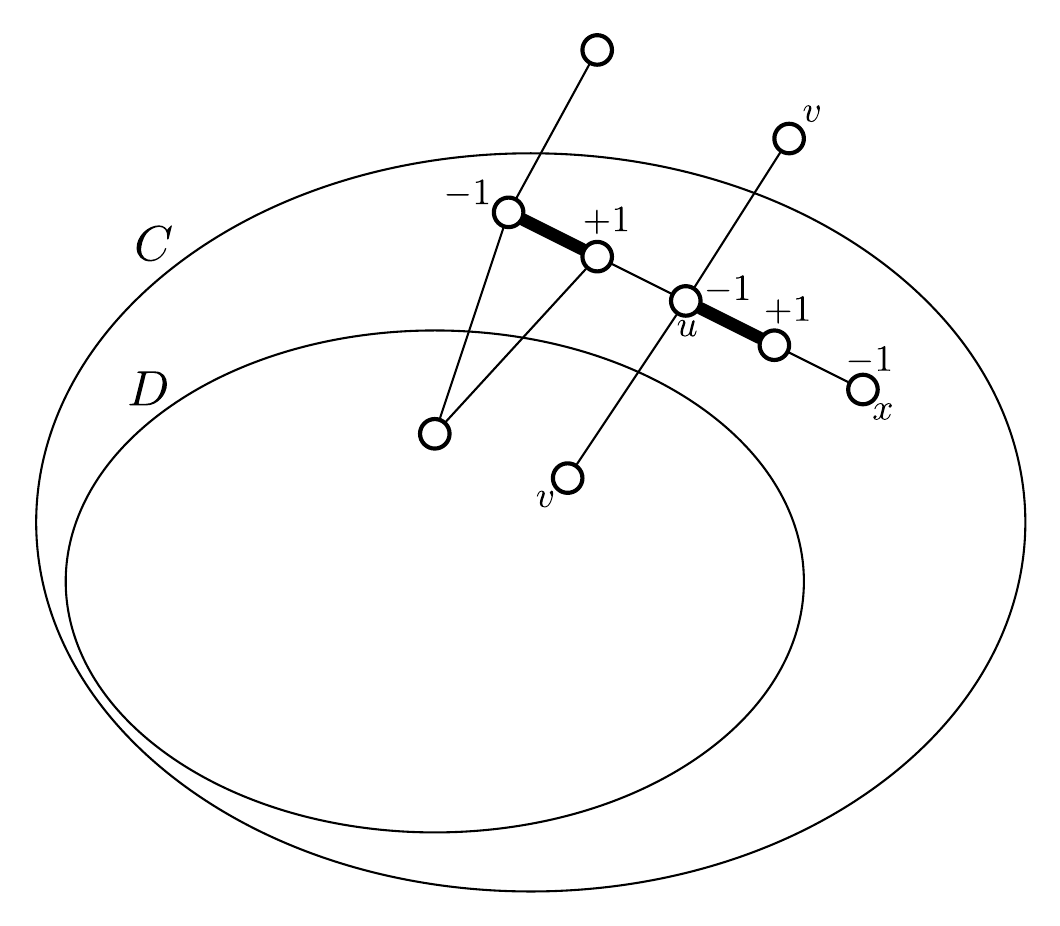}} & \scalebox{.4}{\includegraphics{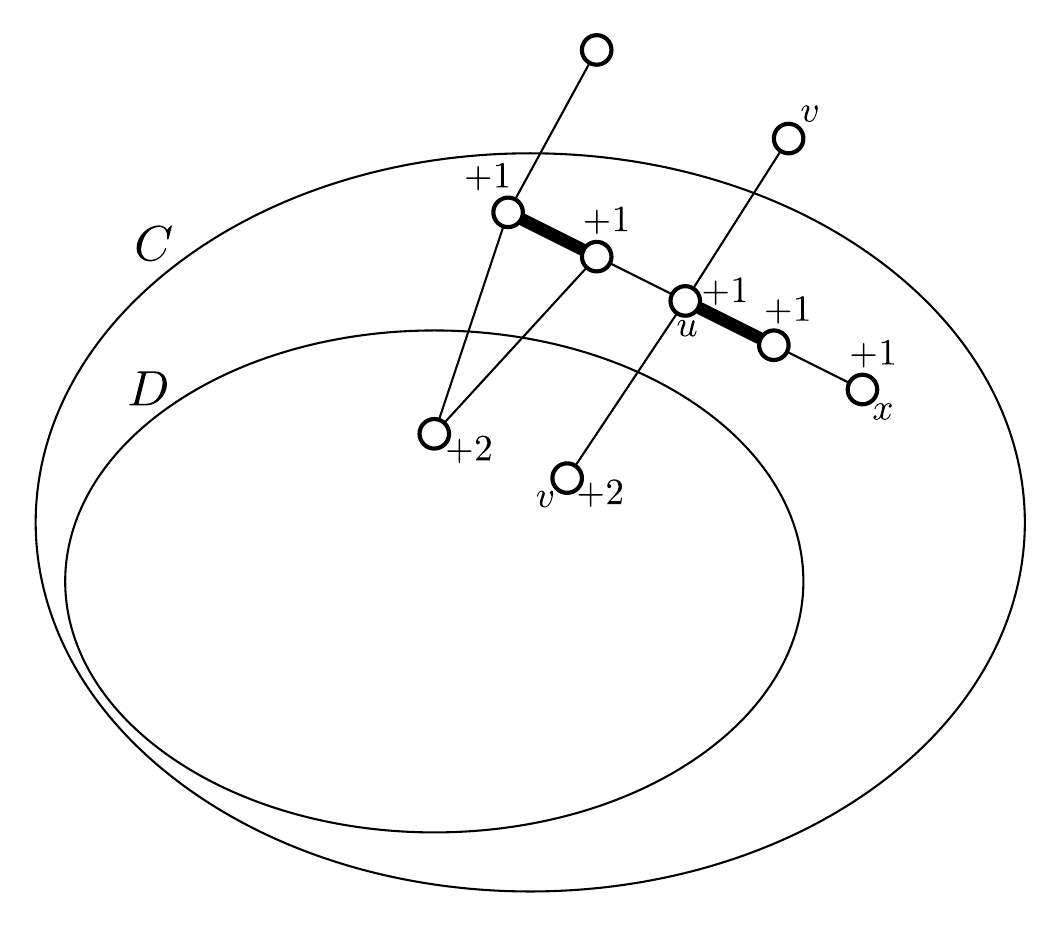}}
\end{tabular}
\caption{\label{fig:Shell}Left: a dual adjustment performed in a shell $G(C,D)$
decrements some $y$-values and may violate Property~\ref{prop:CS} (domination) for edges $(u,v)$ crossing the shell, with $u\in C\backslash D$ and
$v \not\in C$ or $v\in D$.  Right: a unit-translation of $C$ and $D$ decreases $z(C)$ and $z(D)$ by 2 and increases the $y$-values of vertices in $C$ as indicated.
This increases $yz(u,v)$ for each crossing edge $(u,v)$, and preserves domination.}
\end{figure}

Blossom translations are used to preserve Property~\ref{prop:CS}(domination) for all edges, specifically those crossing the shell boundaries.
We implement $\ShellSearch(C,D)$ using an array of buckets for the priority queue, as in \BucketSearch, 
and execute the Augmentation step using a cardinality matching algorithm such as~\cite{MV80,Vazirani12,Vazirani14} or~\cite[\S 10]{GT91} or~\cite{Gabow17}.
Let $t$ be the number of dual adjustments, $G(C^*,D^*)$ be the current atomic shell before the last dual adjustment, and $p\ge 0$ be the number of augmenting paths discovered before halting.  The running time of $\ShellSearch(C,D)$ is $O(t + m(C^*,D^*)\cdot \min\{p+1, \sqrt{n(C^*,D^*)}\})$.
We will show that $t$ is bounded by $O({n}(C^{*},D^{*})))$ as long as the number of free vertices inside $G(C^{*}, D^{*})$ is at least 2.
See Corollary \ref{cor:shellbound}.

The first stage of $\Path(R)$ ends when either all old blossoms in $P(R)$ have dissolved (in which case it halts immediately) or there is exactly one free vertex remaining in an undissolved blossom.
In the latter case we proceed to the second stage of $\Path(R)$ and liquidate all remaining old blossoms.
This preserves Property~\ref{prop:CS} but screws up the dual objective $yz(R)$, which must be corrected before we can halt.
Let $\omega$ be the last free vertex in an undissolved blossom in $R$ and $T = \sum_{i}z(B_i)/2$ be the aggregate amount
of translations performed when liquidating the blossoms.  We perform $\PQSearch(\{\omega\})$, halting after exactly $T$ dual adjustments.
The search is guaranteed not to find an augmenting path.  It runs in $O({m}(R)+{n}(R) \log {n}(R))$ time~\cite{Gabow16}
or $O(m(R)\sqrt{\log\log n(R)})$ w.h.p.; see Section~\ref{sect:ImplementingEdmonds}.

To summarize, $\Path(R)$ dissolves all old blossoms in $P(R)$, either in stage 1, through gradual translations, or in stage 2 through liquidation.  
Moreover, Property 1 is maintained throughout $\Path(R)$.  
In the following, we will show that $\Path(R)$ takes $O({m}(R) ({n}(R))^{3/4})$ time and the dual objective value 
$yz(S)$ does not increase for every $S$ such that $R \subseteq S$. 
In addition, we will show that at all times, the $y$-values of all free vertices have the same parity.

\subsubsection{Properties}

We show the following lemmas to complete the proof of Lemma \ref{lem:gabow}. Let $y_0, z_0$ denote the initial duals, before calling Gabow's algorithm.

\begin{lemma}\label{lem:yvalue} 
After the call to $\Path(R)$, we have $y(u) \geq y_0(u)$ for all $u \in R$. 
Moreover, the $y$-values of free vertices in $R$ are always odd.
\end{lemma}
\begin{proof}

We will assume inductively that this holds after every recursive call of $\Path(R')$ for every $R'$ that is a non-major child of a $P(R)$ node. Then, it suffices to show $y(u)$ does not decrease and the parity of free vertices always stays the same during $\Path(R)$. Consider doing a unit of dual adjustment inside the shell $G(C^{*},D^{*})$. 
Due to the translations of $C^*$ and $D^*$, every vertex in $D^{*}$ has its $y$-value increased by 2
and every vertex in $C^{*}$ either has its $y$-value unchanged or increased by 1 or 2. 
The $y$-values of the free vertices in $C^{*}\backslash D^*$ remain unchanged.  (The dual adjustment decrements
their $y$-values and the translation of $C^*$ increments them again.)

Consider the second stage of \Path$(R)$.  In $\ShellSearch(C,D)$, only augmenting paths within atomic shells can be found, so only the smallest atomic shell can contain odd number of free vertices. 
Therefore $\omega$ is in $B_1$, the smallest undissolved blossom.			
When liquidating blossoms $\{B_i\}$, 
$y(\omega)$ increases by $T \bydef \sum_i z(B_i)/2$ before the call to $\PQSearch(\{\omega\})$. 
Define $w'(u,v) = yz(u,v) - w(u,v)$. The eligible edges must have $w'(u,v) = 0$. 
We can easily see that when we dissolve $B_i$ and increase the $y$-values of vertices in $B_i$, the $w'$-distance from $\omega$ to any vertex outside the largest undissolved blossom $B_\ell$ increases by $z(B_i)/2$. 
Therefore, the total distance from $\omega$ to any vertex outside $B_\ell$ increases by $T$ after dissolving all the blossoms, 
since $\omega\in B_1$.  
Every other vertex inside $B_\ell$ is matched, so
$\PQSearch(\{\omega\})$ will perform $T$ dual adjustments and halt before finding an augmenting path.
We conclude that $y(\omega)$ is restored to the value it had before the second stage of \Path(R).
\end{proof}

\begin{lemma}\label{thm:offend3}
If Property \ref{prop:blossom} holds and 
$y$-values of free vertices have the same parity, 
then Property \ref{prop:blossom} holds after calling $\ShellSearch(C,D)$.
\end{lemma}
\begin{proof}
The current atomic shell $G(C^{*},D^{*})$ cannot contain any old undissolved blossoms, since we are calling $\Path(R)$ in postorder. 
Because we are simulating $\EdmondsSearch(F^*)$ from the set $F^*$ of free vertices in $G(C^{*},D^{*})$,
whose $y$-values have the same parity, by Lemma \ref{lem:offendCS},
Property \ref{prop:blossom} holds within $G(C^{*},D^{*})$. It is easy to check that Property \ref{prop:blossom}(1) (granularity)
holds in $G$. Now we consider Property \ref{prop:blossom}(3,4) (domination and tightness) for the edges crossing $C^{*}$ or $D^{*}$. By Property \ref{prop:blossom}(2c) there are no crossing matched edges and all the newly created blossoms lie entirely in $G(C^{*},D^{*})$. Therefore, tightness must be satisfied. 
The translations on blossoms $C^{*}$ and $D^{*}$ keep the $yz$-values of edges straddling $C^*\backslash D^*$ 
non-decreasing, thereby preserving domination.

Now we claim Property \ref{prop:blossom}(2) holds. We only consider the effect on the creation of new blossoms, since the dissolution of $C^{*}$ or $D^{*}$ cannot violate Property \ref{prop:blossom}(2).  Since edges straddling the atomic shell $G(C^*,D^*)$ are automatically ineligible, we will only create new blossoms inside $G(C^{*}, D^{*})$. 
Since $G(C^{*}, D^{*})$ does not contain any old undissolved blossoms and the new blossoms in $G(C^{*}, D^{*})$ form a laminar set, 
Property \ref{prop:blossom}(2a,2b) holds. 
Similarly, the augmentation only takes place in $G(C^{*}, D^{*})$ which does not contain old undissolved blossoms, Property \ref{prop:blossom}(2c) holds.
\end{proof}

\begin{lemma}\label{lem:yznoincrease} 
The value of $yz(V)$ is non-increasing after the call to $\Path(R)$.	
\end{lemma}

\begin{proof}
Consider a dual adjustment in $\ShellSearch(C,D)$ in which $F^{*}$ is the set of free vertices in the current atomic shell $G(C^{*},D^{*})$. 
By Property~\ref{prop:blossom}(tightness), each dual adjustment within the shell decreases $yz(R)$ by $|F^*|$
since free vertices' $y$-values are decremented and $yz(e)$ is unchanged for each matched edge $e$.
The translation on $C^{*}$ increases $yz(R)$ by 1, 
and if $D^{*} \neq \emptyset$, the translation of $D^{*}$ also increases $yz(R)$ by 1.
Therefore, a dual adjustment in $\ShellSearch$ decreases $yz(R)$
by $|F^{*}|-2$, if $D^{*} \neq \emptyset$, 
and by $|F^{*}|-1$ if $D=\emptyset$.
Since $G(C^{*},D^{*})$ contains at least 2 free vertices, $yz(R)$ does not increase during the first stage of $\Path(R)$.		

Suppose the second stage of \Path$(R)$ is reached, that is, there is exactly one free vertex $\omega$ in an undissolved blossom in $R$.  
When we liquidate all remaining blossoms in $R$,  $yz(R)$ increases by $T$.  
As shown in the proof of Lemma \ref{lem:yvalue}, $\PQSearch(\{\omega\})$ 
cannot stop until it reduces $yz(\omega)$ by $T$. Since Property~\ref{prop:blossom}(tightness) is maintained,
this also reduces $yz(R)$ by $T$, thereby restoring $yz(R)$ back to its value before the second stage of $\Path(R)$.
Since $\Path(R)$ only affects the graph induced by $R$, the arguments above show that $yz(S)$ is non-increasing,
for every $S\supseteq R$.
\end{proof}

The following lemma considers a {\em not necessarily atomic} undissolved shell $G(C,D)$ at some point in time, 
which may, after blossom dissolutions, become an atomic shell.
Specifically, $C$ and $D$ are undissolved but there could be {\em many} undissolved $C'\in \Omega'$ for which $D \subset C' \subset C$.

\begin{lemma}\label{lem:shellbound} 
Consider a call to $\Path(R)$ and any shell $G(C,D)$ in $R$.
Throughout the call to $\Path$, so long as $C$ and $D$ are undissolved 
(or $C$ is undissolved and $D=\emptyset$)
$yz(C) - yz(D) \geq y_0z_0(C) - y_0z_0(D) - 3{n}(C \setminus D)$. 
\end{lemma}

\begin{proof}
If $D = \emptyset$, we let $D'$ be the singleton set consisting of an arbitrary vertex in $C$.
Otherwise, we let $D' = D$.  
Let $\omega$ be a vertex in $D'$. Since blossoms are {\em critical}, we can find a perfect matching $M_{\omega}$ that is also perfect when restricted to $D' \setminus \{ \omega \} $ or $C' \setminus D'$, 
for any $C'\in\Omega'$ with $C' \supset D'$.
($M_{\omega}$ can be derived from the matching $M'$ from the previous scale, by changing the base of $R$ to $\omega$. We only case about the part of $M_{\omega}$ within $C$.)		
By Lemma~\ref{thm:negative}, every $e\in M_\omega \cap E_R$ has 
$y_0z_0(e) \leq w(e) + 6$.  Therefore, 
\begin{align*}
\lefteqn{\sum_{e \in M_{\omega} \cap E(C \setminus D')}w(e) }\\
		&\geq \sum_{e \in M_{\omega} \cap E(C \setminus D')} y_0 z_0(e) -6{n}(C\setminus D')/2 \\
		&= \sum_{u \in V(C \setminus D')} y_0(u) +    \sum_{\substack{C' \in \Omega':\\D' \subset C'}} z_0(C')\cdot \frac{|C' \cap C|- |D'|}{2} 
											+\sum_{\substack{C''\in\Omega':\\ C''\subset C\setminus D'}} z_0(C'') \floor{\frac{C''}{2}} -3{n}(C\setminus D') \\				
		&= y_0 z_0(C) - y_0 z_0(D') - 3{n}(C \setminus D').
\end{align*} 

On the other hand, by Property \ref{prop:blossom}(domination), we have 
\begin{align*}
&\sum_{e \in M_{\omega} \cap E(C\setminus D')}w(e) \\
&\leq \sum_{e \in M_{\omega} \cap E(C \setminus D')} yz(e) \\
&= \sum_{u \in V(C \setminus D')}y(u) + \sum_{\substack{C' \in \Omega':\\D' \subset C'}}z(C')\cdot\frac{|C' \cap C|- |D'|}{2}  + \sum_{B \in \Omega}z(B)\cdot{| M_{\omega}  \cap E(B \cap C\setminus D' )| }\\
&\leq \sum_{u \in V(C \setminus D')}y(u) + \sum_{\substack{C' \in \Omega':\\D' \subset C'}}z(C')\cdot\frac{|C' \cap C|- |D'|}{2}  + \sum_{\substack{B \in \Omega:\\B\subset C}} z(B)\cdot \lfloor \frac{|B| - |B \cap D'|}{2} \rfloor\\ 
\intertext{Consider a $B\in\Omega$ that contributes a non-zero term to the last sum.  By Property~\ref{prop:blossom}, $\Omega\cup\Omega'$ is laminar so
either $B\subseteq D$ or $B\subseteq C\setminus D$.  In the first case $B$ contributes nothing to the sum.  In the second case we have 
$|B\cap D'| \leq 1$ (it can only be 1 when $D=\emptyset$ and $D'$ is a singleton set intersecting $B$) so it contributes exactly $z(B)\cdot \floor{|B|/2}$.
Also, since $\Omega\cup\Omega'$ is laminar and $D' \subset C'$, $C'\cap C$ is either $C'$ or $C$, so $|C' \cap C|$ and $|D'|$ are odd. 		
Continuing on,}
&= \sum_{u \in V(C \setminus D')}y(u) + \sum_{\substack{C' \in \Omega':\\D' \subset C'}} z(C')\cdot\frac{|C' \cap C|- |D'|}{2}  + \sum_{\substack{B \in \Omega:\\B\subset (C\setminus D)}} z(B)\cdot  \lfloor\frac{|B|}{2}\rfloor \\ 
&= yz(C) - yz(D').
\end{align*}
Therefore, $yz(C) - yz(D') \geq y_0z_0(C) - y_0z_0(D') - 3{n}(C\setminus D')$.		
When $D = \emptyset$ we have $yz(D') = y(\omega) \geq y_0(\omega)$. Therefore, regardless of $D$, $yz(C) - yz(D) \geq y_0z_0(C) - y_0z_0(D) - 3{n}(C\setminus D)$.		
\end{proof}

\begin{corollary}\label{cor:shellbound}
The number of dual adjustments in $\ShellSearch(C,D)$ is bounded by $O(n(C^{*} \setminus D^{*}))$ 
where $G(C^{*}, D^{*})$ is the current atomic shell when the last dual adjustment is performed.
\end{corollary}
\begin{proof}
We first claim that the recursive calls to $\Path(R')$ on the descendants $R'$ of $P(R)$ do not decrease $yz(C^{*}) - yz(D^{*})$. If $R' \subset D^{*}$, then any dual adjustments done in $\Path(R')$ changes $yz(C^{*})$ and $yz(D^{*})$ by the same amount. Otherwise, $R' \subset G(C^{*}, D^{*})$. In this case, $\Path(R')$ has no effect on $yz(D^{*})$ and does not increase $yz(C^{*})$ by Lemma \ref{lem:yznoincrease}. Therefore, $yz(C^{*}) - yz(D^{*}) \leq y_0z_0(C^{*}) - y_0z_0(D^{*})$.

First consider the period in the execution of $\ShellSearch(C,D)$ when $D^{*} \neq \emptyset$.
During this period $\ShellSearch$ performs some number of dual adjustments, say $k$.
There must exist at least two free vertices in $G(C^{*},D^{*})$ that participate in all $k$ dual adjustments.  
Note that a unit translation on an old blossom $C''\in\Omega'$, where $D^{*} \subseteq C'' \subseteq C^{*}$, has no net effect on $yz(C^{*}) - yz(D^{*})$, since it increases both $yz(C^{*})$ and $yz(D^{*})$ by 1. 
Thus, each dual adjustment reduces $yz(C^{*}) - yz(D^{*})$ by the number of free vertices in the given shell,
that is, by at least $2k$ over $k$ dual adjustments.  (See the proof of Lemma~\ref{lem:yznoincrease}.)
By Lemma \ref{lem:shellbound}, 
$yz(C^{*})-yz(D^{*})$ decreases by at most $3n(C^{*} \setminus D^{*})$ overall, which implies that $k\le 3/2\cdot n(C^{*} \setminus D^{*})$.

Now consider the period when $D^{*} = \emptyset$.  Let $G(C', D')$ to be the current atomic shell just before the smallest undissolved blossom $D'$ dissolves
and let $k'$ be the number of dual adjustments performed in this period, after $D'$ dissolves.
By Lemma \ref{lem:yznoincrease}, all prior dual adjustments have not increased $yz(C^{*})$.
There exists at least 3 free vertices in $C^*$ that participate in all $k'$ dual adjustments. 
Each translation of $C^*$ increases $yz(C^*)$ by 1.  According to the proof of Lemma \ref{lem:yznoincrease}, 
$yz(C^{*})$ decreases by at least $3k' - k' = 2k'$ due to the $k'$ dual adjustments and translations performed in tandem.  
By Lemma \ref{lem:shellbound}, $yz(C^{*})$ can decrease by at most $3n(C^{*})$, so $k' \le 3/2\cdot n(C^*)$.
The total number of dual adjustments is therefore $k+k' \le 3/2(n(C'\setminus D') + n(C^*)) < 3n(C^*)$.
\end{proof}



The following two lemmas are adapted from~\cite{G85}.  

\begin{lemma}\label{lem:freevertices} 
Let $F$ be the set of free vertices in an undissolved blossom of $P(R)$, at some point in the execution of $\Path(R)$.
For any fixed $\epsilon > 0$, the number of iterations of $\Path(R)$ with $|F| \geq ({n}(R))^{\epsilon}$ is $O(({n}(R))^{1-\epsilon})$.
\end{lemma}

\begin{proof}
Consider an iteration in $\Path(R)$. Let $f$ be the number of free vertices before this iteration. Call an atomic shell {\it big} if it contains strictly more than 2 free vertices. 
We consider two cases depending on whether more than $f/2$ vertices are in big atomic shells or not. Suppose big shells do contain more than $f/2$ free vertices.
The free vertices in an atomic shell will not participate in any dual adjustment only if some adjacent shells have dissolved into it. 
Suppose a shell containing $f'$ free vertices dissolves into (at most 2) adjacent shells and, simultaneously, 
the call to $\ShellSearch$ finds an augmenting path and halts. 
This prevents at most $2f'$ free vertices in the formerly adjacent atomic shells from participating in a dual adjustment, 
since we sorted the shells in non-increasing order by number of free vertices.	
Since there are more than $f/2$ vertices in big atomic shells, 
at least $f/6$ free vertices in the big shells participate in at least one dual adjustment. 
Let $S_i$ be a big even shell with $f_i$ free vertices.  
If they are subject to a dual adjustment then, according to the proof of Lemma \ref{lem:yznoincrease}, 
$yz(R)$ decreases by at least $(f_i - 2)\geq f_i/2$, since the shell is big.  
If $S_i$ is a big {\em odd} shell then the situation is even better.  In this case 
$yz(R)$ is reduced by $(f_i - 1)\geq \frac{2}{3}f_i$.  
Therefore, when $f/2$ free vertices are in big shells, $yz(R)$ decreases by at least $f/12$.

The case when more than $f/2$ free vertices are in small atomic shells can only happen 
$O(\log n)$ times. In this case, there are at least $\lfloor f/4 \rfloor$ small shells. 
In each shell, there must be vertices that were matched during the previous iteration, which implies 
that there must have been at least $f + 2 \lfloor f/4 \rfloor$ free vertices in the previous iteration.  
Thus, we can only be in this situation $\ceil{\log_{3/2} n(R)}$ times, since the number of free vertices shrinks by a 3/2 factor each time.

By Lemma \ref{lem:yznoincrease}, $yz(R)$ does not increase in the calls to $\Path$ on the descendants of $P(R)$. 
By Lemma \ref{lem:shellbound}, since $yz(R)$ decreases by at most $3{n}(R)$, 
the number of iterations with $|F| \geq ({n}(R))^{\epsilon}$ 
is at most $O({n}(R)^{1-\epsilon} +  \log n(R)) = O({n}(R)^{1-\epsilon})$.
\end{proof}

\begin{lemma}
$\Path(R)$ takes at most $O(m(R)(n(R))^{3/4})$ time.
\end{lemma}

\begin{proof}
Recall that $\ShellSearch$ is implemented like $\BucketSearch$, using an array for a priority queue; 
see Section~\ref{sect:ImplementingEdmonds}.
This allows all operations (insert, deletemin, decreasekey)
to be implemented in $O(1)$ time, but incurs an overhead {\em linear} in the number of dual adjustments/buckets scanned.
By Corollary \ref{cor:shellbound} this is $\sum_{i}O(n(S_i)) = O({n}(R))$ per iteration. 
By Lemma \ref{lem:freevertices}, there are at most $O(({n}(R))^{1/4})$ iterations with $|F| \geq ({n}(R))^{3/4}$. Consider one of these iterations. Let $\{S_i\}$ be the shells at the end of the iteration. The augmentation step
takes $\sum_{i}O({m}(S_i)\sqrt{{n}(S_i)})  = O({m}(R)\sqrt{{n}(R)})$ time.
Therefore, the total time of these iterations is $O({m}(R)({n}(R))^{3/4})$. 
There can be at most $({n}(R))^{3/4}$ more iterations afterwards, since each iteration matches at least 2 free vertices.
Therefore, the cost for all subsequent Augmentation steps is $O({m}(R)({n}(R))^{3/4})$.
Finally, the second stage of $\Path(R)$, when there is exactly one free vertex in an undissolved blossom, 
involves a single Edmonds search.  This takes $O({m}(R) + {n}(R) \log {n}(R))$ time~\cite{Gabow16} 
or $O(m(R)\sqrt{\log\log n(R)})$ time w.h.p.; see Section~\ref{sect:ImplementingEdmonds}.
Therefore, the total running time of $\Path(R)$ is $O({m}(R)({n}(R))^{3/4})$.
\end{proof}

Let us summarize what has been proved.  By the inductive hypothesis, all calls to $\Path$ preceding $\Path(R)$ have (i) dissolved all old blossoms in $R$ excluding those in $P(R)$, (ii) kept the $y$-values of all free vertices in $R$ the same parity (odd) and kept $yz(R)$ non-increasing, and (iii) maintained Property \ref{prop:blossom}.  If these preconditions are met, the call to $\Path(R)$ dissolves all remaining old blossoms in $P(R)$ while satisfying (ii) and (iii).
Futhermore, $\Path(R)$ runs in $O(m(R) (n(R))^{3/4} )$ time. This concludes the proof of Lemma~\ref{lem:gabow}.

\section{Implementing Edmonds' Search}\label{sect:ImplementingEdmonds}

This section gives the details of a reasonably efficient implementation of Edmonds' search.
Previous algorithms for real-weighted inputs, such as Galil et al.'s~\cite{GalilMG86} and Gabow's~\cite{Gabow16},
implement specialized priority queues for dealing with blossom formulation/dissolution.  These data structures
do not benefit from having \emph{integer}-valued duals.  Indeed, their per-operation running times are $\Omega(\log n)$ for 
reasons that have nothing to do with the $n\log n$ lower bound on comparison-based sorting.

The implementation
of Edmonds' algorithm presented here was suggested by Gabow~\cite{G85}.  
It uses an ``off the shelf'' priority queue (among other data structures),
and can therefore be sped up when the graph happens to be integer-weighted.
When the duals are integers and the number
of dual adjustments is $t$ it runs in $O(m+t)$ time using a bucket array for the priority queue; this is called \BucketSearch.
When the number of dual adjustments is unbounded we call it 
\PQSearch; it runs in $O(mq)$ time, given a priority queue supporting insert and delete-min in $O(q)$ amortized time.

Let us first walk through a detailed execution of the search for an augmenting path, which
illustrates some of the unusual data structural challenges of implementing Edmonds' algorithm.  
In Figure~\ref{fig:detailed-example} edges are labeled by their initial slacks and blossoms are labeled
by their initial $z$-values; we are performing a search from the set $F=\{u\}$. 
All matched and blossom edges are tight and we are using Criterion~\ref{crit1} (tightness)
for eligibility.
It is convenient to conflate the number of units
of dual adjustment performed by Edmonds' algorithm with {\em time}.  

\begin{figure}[h]
\centering
\begin{tabular}{c@{\hcm[1]}c}
\scalebox{.45}{\includegraphics{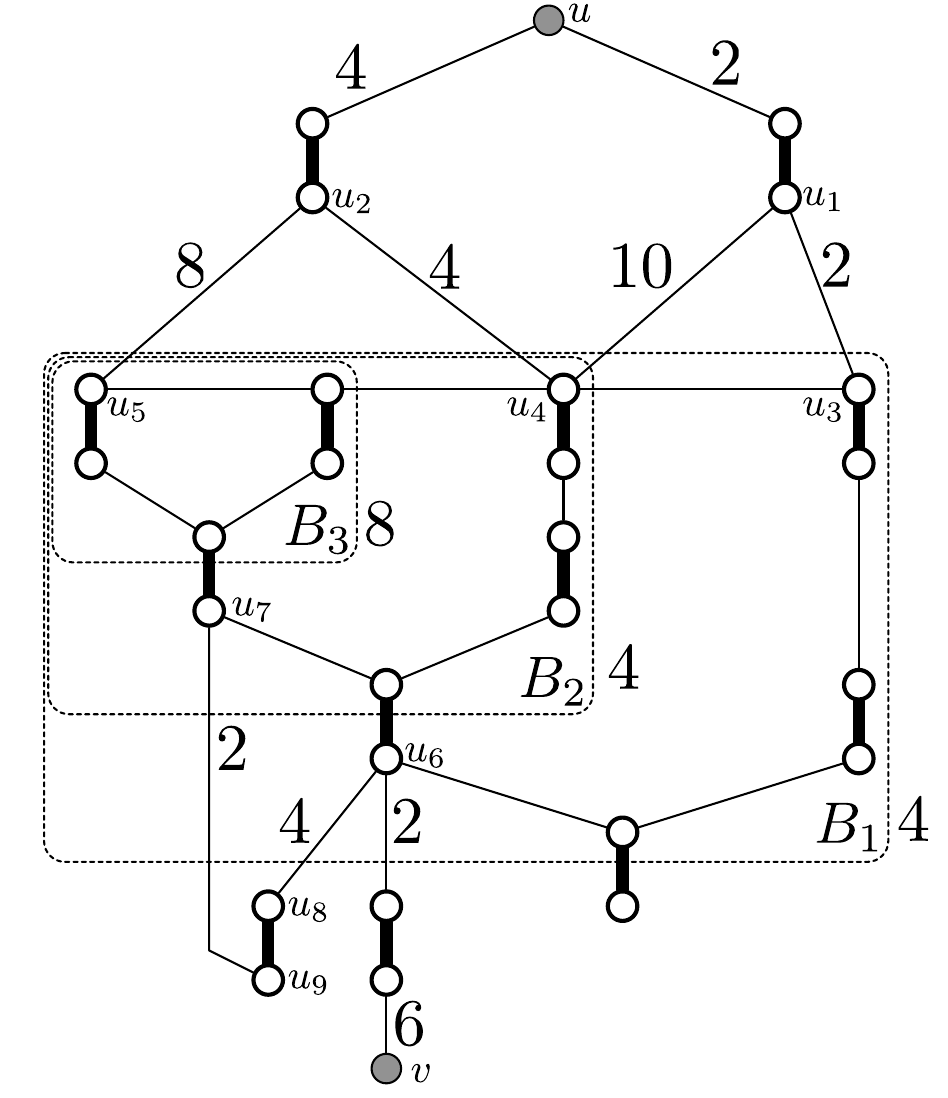}}
&
\scalebox{.45}{\includegraphics{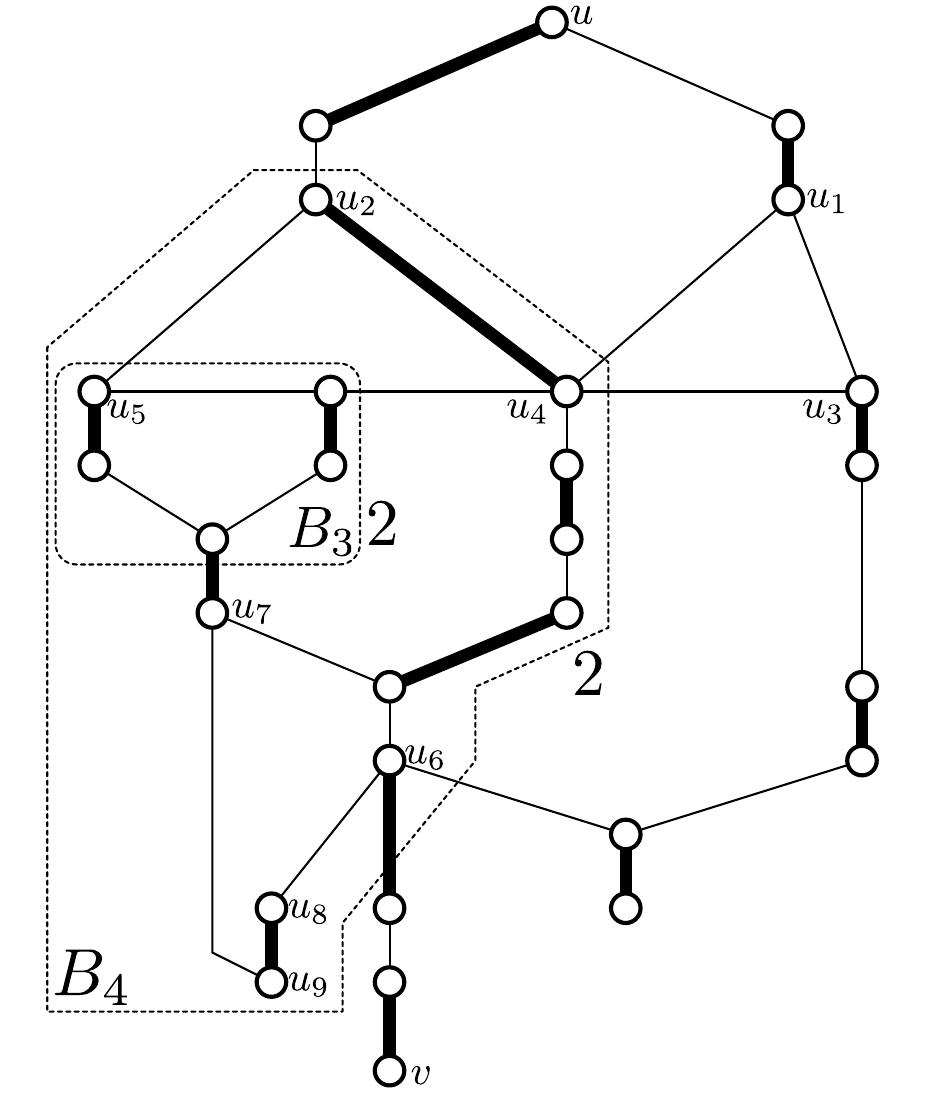}}\\
&\\
\multicolumn{2}{c}{\scalebox{.5}{\includegraphics{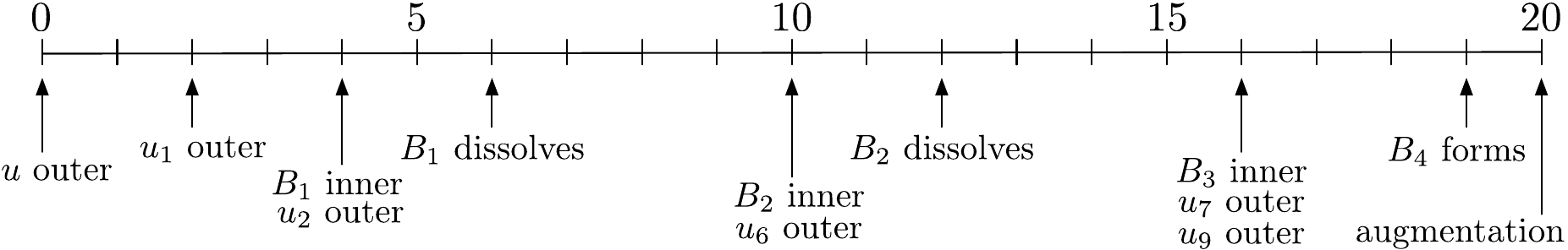}}}
\end{tabular}
\caption{\label{fig:detailed-example}Edges are labeled with their initial slack; blossoms are labeled
with their $z$-values.}
\end{figure}

At time zero $u$ is outer and all other vertices are not in the search structure.  

At time 2 $u_1$ becomes outer and the edges $(u_1,u_3)$ and $(u_1,u_4)$ are scanned.
Both edges connect to blossom $B_1$ and $(u_1,u_3)$ has less slack, but, as we shall see, 
$(u_1,u_4)$ cannot be discarded at this point.  

At time 4 $B_1$ becomes an inner blossom
and $u_2$ becomes outer, causing $(u_2,u_4)$ and $(u_2,u_5)$ to be scanned.
Note that since $u_4$ is inner (as part of $B_1$), further dual adjustments 
will not change the slack on $(u_2,u_4)$ (slack 4) or $(u_1,u_4)$ (now slack 8).
Nonetheless, in the future $u_4$ may not be in the search structure, so we note that the edge
with least slack incident to it is $(u_2,u_4)$ and discard $(u_1,u_4)$.

At time 6 $B_1$ dissolves: the path from $u_3$ to $B_1$'s base enters the search structure
and everything else ($B_2$ and $u_6$) are split off.  At this point further dual adjustments {\em do}
change the slack on edges incident to $B_2$.

At time 10 $(u_2,u_4)$ becomes tight, $B_2$ becomes inner and $u_6$ becomes outer.
At time 12 $B_2$ dissolves.  At this point $u_6$ has experienced 2 dual adjustments as an inner vertex 
(as part of $B_1$) and 2 dual adjustments as an outer vertex, while
$u_7$ has experienced 4 dual adjustments as an inner vertex (as part of $B_1$ and $B_2$).
Thus, the slacks on $(u_6,u_8)$ and $(u_7,u_9)$ are 4 and 6 respectively.

At time 16 $(u_2,u_5)$ and $(u_6,u_8)$ become tight, making 
$B_3$ inner and $u_7$ and $u_9$ outer.
The edge $(u_7,u_9)$ still has slack 6.
At time 19 $(u_7,u_9)$ becomes tight, forming a new blossom $B_4$ based at $u_2$. 
At time 20 the final edge on the augmenting path from $u$ to $v$ becomes tight.

Observe that a vertex can enter into and exit from the search structure an unbounded number
of times.  Merely calculating a vertex's current $y$-value requires that we consider the entire history
of the vertex's involvement in the search structure.  For example, $u_5$ participated as an inner vertex in dual adjustments
during the intervals $[4,6), [10,12),$ and $[16,19)$ and as an outer vertex during $[19,20)$.

\subsection{An Overview of the Data Structures}

In order to implement Edmonds' algorithm efficiently we need to address three data structuring problems:

\begin{enumerate}

\item a {\em union-find} type data structure for maintaining the (growing) outer blossoms.  This data structure is used
to achieve two goals.  First, whenever
an outer-outer edge $e=(u,u')$ is scanned (like $(u_7,u_9)$ in the example), we need to tell whether $u,u'$
are in the same outer blossom, in which case $e$ is ignored, or whether they are in different blossoms, in which 
case we must schedule a blossom formation event after $\slack(e)/2$ further dual adjustments.  Second, 
when forming an outer blossom, we need to traverse its odd cycle in time proportional to its length {\em in the contracted graph}.
E.g., when $(u_7,u_9)$ triggers the formation of $B_4$, we walk up from $u_7$ and $u_9$ to the base $u_2$ enumerating
the vertices/root blossoms encountered.  The walk must ``hop''
from $u_7$ to $u_5$ (representing $B_3$), without spending time proportional to $|B_3|$. 

\item a {\em split-findmin} data structure for maintaining the (dissolving) inner blossoms.  The data structure
must be able to dissolve an inner blossom into the components along its odd cycle. 
It must be able to determine the edge with minimum slack connecting an inner blossom 
to an outer vertex, and to do the same for individual vertices in the blossom.  For example, 
when $(u_2,u_4)$ is scanned we must check whether its slack is better than the 
other edges incident to $u_4$, namely $(u_1,u_4)$.

\item a {\em priority queue} for scheduling three types of events: {\em blossom dissolutions}, {\em blossom formations},
and {\em grow} steps, which add a new (tight) edge and vertex to the search structure.\footnote{Augmenting paths are
discovered in the course of processing a {\em blossom formation} step or {\em grow} step, depending on whether
both ends or just one end of the augmenting path is in $F$.  
In particular, \emph{blossom formation} events record an outer-outer type edge to be processed, whereas a \emph{grow} 
step records an edge to be processed with one endpoint having no inner/outer type.}

\end{enumerate}

Before we get into the implementation details let us first make some remarks on the existing options for (1)--(3).

A standard union-find algorithm will solve (1) in $O(m\alpha(m,n)+n)$ time.  Gabow and Tarjan~\cite{GT85} observed
that a special case of union-find can be solved in $O(m+n)$ time if the data structure gets commitments on future
union operations.  Let $\{\{1\},\{2\},\ldots,\{n\}\}$ be the initial set partition and 
$T = \emptyset$ be an edge set on the vertex set $\{1,\ldots,n\}$.  We must maintain the invariant that $T$ is a single
connected tree at all times.
 The data structure handles intermixed sequences of three operations.
\begin{itemize}
\item[] $\UFaddedge(u,v)$ : $T \leftarrow T\cup\{(u,v)\}$.\\
{\em It is required that $T$ is connected. If $v$ was previously not in $T$, record $u$ as the parent of $v$.}

\item[] $\UFunite(u,v)$ : Replace the sets containing $u$ and $v$ with their union.\\
{\em This is only permitted if $(u,v)\in T$.}

\item[] $\UFfind(u)$ : Find the representative of $u$'s set.
\end{itemize}
It is not too difficult to cast Edmonds' search in this framework.  We explain exactly how in Section~\ref{sect:Edmonds-implementation-details}.

Gabow introduced the split-findmin structure in~\cite{G85} to manage blossom dissolutions in Edmonds' algorithm,
but did not fully specify how it should be applied.  The data structure maintains a set 
$\mathcal{L}$ of lists of elements, each associated with a key.  It supports the following operations.
\begin{itemize}
\item[] $\SFinit(u_1,\ldots,u_n)$ : Set $\mathcal{L} \leftarrow \{(u_1,\ldots,u_n)\}$ and $\key(u_i) \leftarrow \infty$ for all $i$.
\item[] $\SFlist(u)$ : Return a pointer to the list in $\mathcal{L}$ containing $u$.
\item[] $\SFsplit(u)$ : Suppose $\SFlist(u) = (u',\ldots,u,u'',\ldots,u''')$.  Update $\mathcal{L}$ as follows:
				\[
				\mathcal{L} \leftarrow \mathcal{L} \,\backslash\, \{\SFlist(u)\} \cup 
				\{(u',\ldots,u), \; (u'',\ldots,u''')\}.\]
\item[] $\SFdeckey(u,x)$ : Set $\key(u) \leftarrow \min\{\key(u),x\}$.
\item[] $\SFfindmin(L\in\mathcal{L})$ : Return $\min_{u \in L} \key(u)$.
\end{itemize}
The idea is that $\SFinit$ should be called with a permutation of the vertex set such 
that each initial blossom (maximal or not) is contiguous in the list.  Splits are performed
whenever necessary to maintain the invariant that non-outer root blossoms
are identified with lists in $\mathcal{L}$.  The value $\key(u)$ is used to encode the 
minimum slack of any edge $(v,u)$ ($v$ outer) incident to $u$.
We associate other useful information with elements and lists; for example, $u$ stores
a pointer to the edge $(v,u)$ corresponding to $\key(u)$. 

Pettie~\cite{Pet15-sf} improved the running time of Gabow's split-findmin
structure from $O(m\alpha(m,n)+n)$ to $O(m\log\alpha(m,n)+n)$, 
$m$ being the number of $\SFdeckey$ operations.
Thorup~\cite{Tho99} 
showed that with integer keys, split-findmin could be implemented in optimal $O(m+n)$ time using atomic heaps~\cite{FW94}.

For (3) we can use a standard priority queue supporting insert and deletemin.  Note however, that although 
there are ultimately only $O(n)$ events, we may execute $\Theta(n+m)$ priority operations.  The algorithm may schedule
$\Omega(m)$ blossom formation events but, when each is processed, discover that the 
endpoints of the edge in question have already been contracted into the same outer blossom.  
(A decreasekey operation, if it is available, is useful for rescheduling {\em grow} events but cannot directly help with blossom formation
events.) 
Gabow's specialized priority queue~\cite{Gabow16} schedules all blossom formation events in $O(m + n\log n)$ time. 
Unfortunately, the $\Omega(n\log n)$ term in Gabow's data structure cannot be reduced if the edge weights happen
to be small integers.  Let $t_{\max}$ be the maximum number of dual adjustments performed by a search.  
In the $\BucketSearch$ implementation we shall allocate an array of $t_{\max}$ buckets 
to implement the priority queue, bucket $i$ being a linked list of events scheduled for time $i$.
With this implementation all priority queue operations take $O(1)$ time, plus $O(t_{\max})$ for scanning empty buckets.
When $t_{\max}$ is unknown/unbounded we use a general integer priority queue~\cite{Han02,HT02,Thorup07b} 
and call the implementation $\PQSearch$.

In the remainder of this section we explain how to implement Edmonds' search procedure using the data structures 
mentioned above.  This is presumably close to the implementation that Gabow~\cite{G85} had in mind, but it is quite
different from the other $\tilde{O}(m)$ implementations of~\cite{GalilMG86,GGS89,Gabow16}.
Theorem~\ref{thm:Edmonds-search} summarizes the properties of this implementation.

\begin{theorem}\label{thm:Edmonds-search}
The time to perform Edmonds' search procedure 
on an integer-weighted graph, 
using specialized union-find~\cite{GT85},
split-findmin~\cite{Tho99}, and priority queue~\cite{Han02,HT02,Thorup07b} data structures,
is $O(m+t)$ (where $t$ is the number of dual adjustments, using a trivial priority queue)
or $O(m\log\log n)$ (using~\cite{Han02,Thorup07b})
or $O(m\sqrt{\log\log n})$ with high probability (using~\cite{HT02,Thorup07b}).
On real-weighted graphs the time is $O(m+n\log n)$ using~\cite{Gabow16}, 
or $O(m\log n)$ using any $O(\log n)$-time priority queue.
\end{theorem}

\subsection{Implementation Details}\label{sect:Edmonds-implementation-details}

We explicitly maintain the following quantities, for each $v$ and each blossom $B$.  
A vertex $B$ not in any blossom is considered a root blossom, trivially.
\begin{align*}
t_{\now} 	&= \mbox{The current {\em time}. (The number of dual adjustments performed so far.)}\\
y_0(v) 	&= \mbox{The initial value of $y(v)$.}\\
z_0(B)	&= \mbox{The initial value of $z(B)$ ($B$ not necessarily a root blossom).}\\
t_{\root}(B) &= \mbox{The time $B$ became a root blossom.}\\
t_{\Inner}(B)	&= \mbox{The time $B$ became an inner root blossom.}\\
t_{\Outer}(B)	&= \mbox{The time $B$ became (part of) an outer root blossom.}\\
\Delta(B)	&= \mbox{The number of dual adjustments experienced by vertices in $B$ as inner vertices,}\\
		&\;\;\;\;\;\; \mbox{ in the interval $[0,\max\{t_{\root}(B), t_{\Outer}(B)\}]$.}
\intertext{It is straightforward to keep these values up to date.  To give a sense of what is involved, we illustrate how they change in two cases:
when an inner blossom dissolves and when an outer blossom is formed.
Whenever an inner blossom $B'$ is dissolved we visit each subblossom $B$ on its odd-cycle and 
set }
t_{\root}(B) &\leftarrow t_{\now}\\
\Delta(B) &\leftarrow \Delta(B') + (t_{\now} - t_{\Inner}(B'))
\intertext{and if $B$ is immediately inserted into the search structure as an inner or outer blossom 
we set $t_{\Inner}(B)\leftarrow t_{\now}$ or $t_{\Outer}(B)\leftarrow t_{\now}$ accordingly.  
When an outer blossom $B'$ is created we visit each subblossom $B$ on its odd-cycle.
For each formerly inner $B$, we update its values as follows}
t_{\Outer}(B) &\leftarrow t_{\now}\\
\Delta(B) &\leftarrow \Delta(B) + (t_{\now}-t_{\Inner}(B))
\intertext{From these quantities we can calculate the current $y$- and $z$-values as follows.
Remember that $\SFsplit$s are performed so that $\SFlist(v)=B$ was the last root blossom containing $v$
just before $v$ became outer, or the current root blossom containing $v$ if it is non-outer.}
y(v) &= \zero{y_0(v)}\;\;\;\;\;\;\;\;\;
		+ \left\{\begin{array}{l@{\hcm[.3]}l}
			\zero{\Delta(B) + t_{\now} - t_{\Inner}(B)}			& \mbox{ if $B = \SFlist(v)$ is an inner root blossom}\\
			 \zero{\Delta(B) - (t_{\now} - t_{\Outer}(B))}			& \mbox{ if $B = \SFlist(v)$ is in an outer blossom}\\
			 \zero{\Delta(B)}\hcm[5]							& \mbox{ otherwise}
		\end{array}\right.\\
z(B) &= \zero{z_0(B)}\;\;\;\;\;\;\;\;\; + \left\{\begin{array}{l@{\hcm[.3]}l}
				\zero{0}								& \mbox{ if $t_{\root}(B)$ is undefined}\\
				\zero{- 2\Delta(B) 	- 2(t_{\now} - t_{\Inner}(B))} & \mbox{ if $B$ is an inner root blossom}\\
				\zero{- 2\Delta(B) + 2(t_{\now} - t_{\Outer}(B))}\hcm[5]	& \mbox{ if $B$ is in an outer blossom}
		\end{array}\right.
\intertext{Note that if $B$ was not a weighted blossom at time zero, $z_0(B)=\Delta(B)=0$.
The slack of an edge $(u,v)$ not in any blossom is calculated as $\slack(u,v) = y(u) + y(v) - w(u,v)$.
However, when using Criteria~\ref{crit2} or \ref{crit3} of eligibility we really want to measure the distance
from the edge being {\em eligible}.  Define $\slack^\star(u,v)$ as follows.}
\slack^\star(e) &= \left\{
\begin{array}{ll}
\slack(e)				& \mbox{ Criterion~\ref{crit1}}\\
\slack(e)	 + 2			& \mbox{ Criterion~\ref{crit2} and $e\not\in M$}\\
-\slack(e)				& \mbox{ Criterion~\ref{crit2} and $e\in M$}\\
\slack(e)				& \mbox{ Criterion~\ref{crit3} and $\slack(e) \ge 0$}\\
\slack(e)    + 2			& \mbox{ Criterion~\ref{crit3} and $\slack(e) \in \{-1,-2\}$}
\end{array}
\right.
\end{align*}

Dual adjustments can change the slack of many edges, but we can only afford to update the split-findmin structure
when edges are scanned.  We maintain the invariant that if $u$ is not in an outer blossom, 
$\key(u)$ is equal to $\min_{\operatorname{outer} \; v} \slack^\star(v,u)$, up to some offset that is common to all 
vertices in $\SFlist(u)$.  Consider an edge $(v,u)$ with $v$ outer and $u$ non-outer.   When $u$ is not in the search structure
each dual adjustment reduces the slack on $(v,u)$ whereas when $u$ is inner each dual adjustment has no effect.
We maintain the following invariant for each non-outer element $u$ in the split-findmin structure.

\begin{align*}
\min_{\operatorname{outer} \; v} \slack^\star(v,u) &= 
\left\{
	\begin{array}{l@{\hcm[1]}l}
		\key(u) - (t_{\Inner}(B) - \Delta(B))	& \mbox{ if $B=\SFlist(u)$ is inner}\\
		\key(u) - (t_{\now} - \Delta(B))			& \mbox{ if $B=\SFlist(u)$ is neither inner nor outer}
	\end{array}
\right.
\end{align*}

Let $F$ be the set of free vertices that we are conducting the search from.  In accordance with our earlier
assumptions we assume that $\{y_0(v) \;|\; v\in F\}$ have the same parity and that all edge weights are even.
We will grow a forest $T$ of $|F|$ trees, each rooted at an $F$-vertex, such that the outer blossoms form
connected subtrees of $T$, thereby allowing us to apply the union-find algorithm~\cite{GT85} to each tree.
Let $\rt(u)$ be the free vertex at the root of $u$'s tree.
We initialize the split-findmin structure to reflect the structure of initial blossoms
at time $t_{\now} = 0$ and call $\grow(v,\bottom)$ for each $v\in F$.  In general, we iteratively process
any events scheduled for $t_{\now}$, incrementing $t_{\now}$ when there are no such events.  
Eventually an augmenting path will be discovered (during the course of processing a {\em grow} or {\em blossom formation} event)
or the priority queue becomes empty, in which case we conclude
that there are no augmenting paths from any vertices in $F$.

\paragraph{The $\grow(v,e)$ procedure.} 
The first argument ($v$) is a new vertex to be added to the search structure.  The second argument $e=(u,v)$
is an edge with $\slack^\star(e)=0$ connecting $v$ to an existing $u$ in the search structure, or $\bottom$ if $v$ is free.
If $e\neq \bottom$ we begin by calling $\UFaddedge(e)$.

We first consider the case when $e\in M$ or $e=\bottom$, so $v$ is designated {\em outer}.  If $v$ is not contained
in any blossom we call $\schedule(v)$ to schedule {\em grow} and {\em blossom formation} events for all unmatched 
edges incident to $v$.  If $B=\SFlist(v)$ is a non-trivial (outer) blossom we call $\UFaddedge(v',v)$ and $\UFunite(v',v)$,
for each $v'\in B\backslash\{v\}$, then call $\schedule(v')$ for each $v'\in B$.  (Recall that in order to apply~\cite{GT85},
the members of every outer blossom must form a contiguous subset of $T$.)

Suppose $e=(u,v)\not\in M$ and that $v$ is not contained in any blossom.
If $v$ is free then we have found an augmenting path and are done.
Otherwise we call $\schedule(v)$ to schedule the grow step for $v$'s matched edge.\footnote{Under Criterion~\ref{crit1} this would always happen immediately, but under the other Criteria it could happen after 0, 1, or 2 dual adjustments.}
If $B=\SFlist(v)$ is a non-trivial (inner) blossom, find the base $b$ of $B$ and the even-length path $P$ from $v$ to $b$ in 
$E_B$, in $O(|P|)$ time.\footnote{The data structures involved in generating even-length paths through blossoms and finding the current base are well understood.  See Gabow~\cite{Gabow76}, for example.}
For each edge $e\in P$ call $\UFaddedge(e)$ in order to include $P$ in $T$.\footnote{The idea here is to include the minimal 
portion of $E_B$ necessary to ensure connectivity.  The rest of $B$ cannot be included in $T$ yet because parts of it may break
off when $B$ is dissolved.}  If $b$ is free then we have found an augmenting path; if not
then we call $\schedule(b)$ to schedule $B$'s blossom dissolution event and the grow event
for $b$'s matched edge.

\paragraph{The $\schedule(u)$ procedure.}
The purpose of this procedure is to schedule future events associated with $u$ or edges incident to $u$.  
First consider the case when $u$ is inner.  If $B=\SFlist(u)$ is a non-trivial blossom we schedule a $\dissolve(B)$ event at
time $t_{\now} + z_0(B)/2$.  Let $e=(u,v)\in M$ be the matched edge incident to $u$. 
If $v$ is neither inner nor outer then schedule a $\grow(v,e)$ event at time $t_{\now} + \slack^\star(e)$.
If $v$ is currently inner we cancel the existing event for $\grow(u,e)$ and schedule a 
$\blossom(e)$ event at time 
$t_{\now} + \slack^\star(e)/2$.\footnote{Recall that all $y$-values of $F$-nodes have the same parity,
and that any nodes reachable from an $F$-node by a path of eligible edges also have the same parity.
If $u$ and $v$ have both been reached, then $\slack(u,v)$ and $\slack^\star(u,v)$ are both even, since $y(u) + y(v)$ and $w(u,v)$ are both even.}

When $u$ is outer we perform the following steps for each unmatched edge $e=(u,v)\in E(G)$.
If $\UFfind(u)=\UFfind(v)$ then $(u,v)$ can be discarded.
If $v$ is also outer and $\UFfind(u)\neq\UFfind(v)$ 
then schedule a $\blossom(e)$ event at time $t_{\now} + \slack^\star(e)/2$.
If $v$ is inner let $(u',v)$ be the existing edge with $u'$ outer minimizing $\slack^\star(u',v)$.
If $\slack^\star(e) < \slack^\star(u',v)$ then we perform a $\SFdeckey(v,x)$ operation with the new key $x$
corresponding to $\slack^\star(e)$.  If $B=\SFlist(v)$ is neither inner nor outer and updating $\key(v)$
causes $\SFfindmin(B)$ to change, we cancel the existing grow event associated with $B$
and schedule $\grow(v,e)$ for time $t_{\now} + \slack^\star(e)$.

\paragraph{The $\dissolve(B)$ procedure.}
Let $P = T\cap E_B$ be the even-length alternating path from some $v\in B$ to the  base $b$ of $B$. 
For each subblossom $B'$ on $B$'s odd cycle we call $\SFsplit(u')$ on the last
vertex $u'\in B$, thereby splitting $B$ into its constituents.  The subblossoms $B'$ are of three kinds: they  
either 
(i) intersect $P$ as inner vertices/blossoms, 
(ii) intersect $P$ as outer vertices/blossoms,
or
(iii) do not intersect $P$.
If $B'$ is of type (i) then subsequent dual adjustments will reduce $z(B')$.  We schedule a $\dissolve(B')$
event for time $t_{\now} + z_0(B')/2$.  If $B'$ is type (ii) let $b'$ be its base.  Every vertex $v' \in B'$ is now 
outer.  For each $v'\in B'\backslash\{b'\}$ we call $\UFaddedge(v',b'), \UFunite(v',b')$ and
for each $v'\in B'$ we call $\schedule(v')$ to schedule events for unmatched edges incident to $v'$.
When $B'$ is type (iii) we call $\SFfindmin(B')$ to determine the unmatched edge $e=(u,v)$ ($u$ outer, $v\in B'$)
minimizing $\slack^\star(u,v)$.  We schedule a $\grow(v,e)$ event at time $t_{\now} + \slack^\star(u,v)$.

\paragraph{The $\blossom(u,v)$ procedure.}
When a $\blossom(u,v)$ event occurs either $\slack^\star(u,v)=0$
or $u$ and $v$ have already been contracted into a common outer blossom.
If $\UFfind(u)=\UFfind(v)$ we are done.
If $\rt(u)\neq \rt(v)$ then we have discovered an eligible augmenting path from $\rt(u)$ to $\rt(v)$ via $(u,v)$. 
If $\rt(u)=\rt(v)$ then a new blossom $B$ must be formed.  The base of $b$ will be the least common ancestor
of $u$ and $v$ in $T$.  We walk from $u$ up to $b$ and from $v$ up to $b$, making sure that all members
of $B$ are in the same set defined by $\UFunite$ operations.  Here we must be more specific about which vertex in
a blossom is the ``representative'' returned by $\UFfind(\cdot)$.  The representative of a blossom is its most ancestral node
in $T$. For outer blossoms this is always the base; for inner blossoms this is the vertex $v$ in the call
to $\grow(v,(u,v))$ that caused $v$'s blossom to become inner.   

Let $u'$ be the current vertex under consideration
on the path from $u$ to $b$.   If $u'$ is outer, in a non-trivial blossom, but not the base of the blossom,
set $u' \leftarrow \UFfind(u')$ to be the base of the blossom and continue.  Suppose $u'$ is the base of an outer
blossom and $v'$ is its (inner) parent.  Call $\UFunite(u',v')$; set $u'\leftarrow v'$ and continue.
If $u'$ is inner, not in any blossom, and $v'$ is its parent, call $\UFunite(u',v')$; set $u'\leftarrow v'$ and continue.
Suppose $u'$ is in a non-trivial inner blossom $B'=\SFlist(u')$.
Let $P'=E_{B'} \cap T$ be the (possibly empty) path from $u'$ to the representative $v'$ of $B'$ and let $u''$ be the parent of $v'$.
Call $\UFunite(e)$ for each $e\in P'$ then, for each $v'' \in B' \backslash V(P')$, call $\UFaddedge(v'',v')$ and $\UFunite(v'',v')$. 
Call $\UFunite(v',u'')$; set $u'\leftarrow u''$ and continue.  The same procedure is repeated on the path from $v$ up to $b$.
Note that the time required to construct $B$ is linear in the number of $B$-vertices that make the transition from inner to outer.
Thus, the {\em total} time for forming all outer blossoms is $O(n)$.

For each $v'\in B$ that was not already outer before the formation of $B$, call $\schedule(v')$ to schedule events
for unmatched edges incident to $v'$.

\subsection{Postprocessing}

Once a single augmenting path is found we explicitly record all $y$- and $z$-values, in $O(n)$ time.
At this moment the (relaxed) complementary slackness invariants (Property~\ref{prop:CS} or~\ref{prop:RCS})
are satisfied, except possibly the {\em Active Blossom} invariant.  Any blossoms that
were formed at the same time that the first augmenting path was discovered will have zero $z$-values.
Also, a non-root blossom with zero $z$-value may become a root blossom just as the first augmenting path is found.
Thus, we must dissolve root blossoms with zero $z$-values as long as they exist.

\section{Conclusion}\label{sect:conclusion}

We have presented a new scaling algorithm for \MWPM{} on general graphs that runs in $O(m\sqrt{n}\log(nN))$ time.
This algorithm improves slightly on the running time of the Gabow-Tarjan algorithm~\cite{GT91}.  However, its analysis
is simpler than~\cite{GT91} and is generally more accessible.  
Historically there were two barriers to computing weighted matching in less than $O(m\sqrt{n}\log(nN))$ time.
The first barrier was that the best cardinality matching algorithms took $O(m\sqrt{n})$ time~\cite{Vazirani12,Vazirani14,GT91,Gabow17}, 
and cardinality matching seems easier than a single scale of weighted matching.
The second barrier was that even on {\em bipartite} graphs, where blossoms are not an issue, the best matching algorithms 
took $O(m\sqrt{n}\log(nN))$ time~\cite{GT89,OrlinA92,GoldbergK97,DuanS12}.
Recent work by Cohen, \Madry, Sankowski, and Vladu~\cite{CohenMSV17} has broken the second barrier on sufficiently sparse graphs.
They showed that several problems, including weighted bipartite matching, can be computed in $\tilde{O}(m^{10/7}\log N)$ time.

We highlight several problems left open by this work.
\begin{itemize}
\item The \Liquidationist{} \MWPM{} algorithm is relatively simple and streamlined, 
and among the scaling algorithms for \MWPM{} so-far proposed~\cite{G85,GT91}, the one with the clearest potential for practical impact.
However, on sparse graphs it is theoretically an $O(\sqrt{\log\log n})$ factor slower than the \Hybrid{} algorithm.  Can the efficiency of \Hybrid{} be matched by an algorithm that is as simple as \Liquidationist?

\item There is now some evidence that the maximum weight (not necessarily perfect) matching problem~\cite{DuanS12,Pettie12,HuangK12,KaoLST01} may be slightly easier than \MWPM.  Is it possible to compute a maximum weight matching
of a general graph in $O(m\sqrt{n}\log N)$ time, matching the bound of Duan and Su~\cite{DuanS12} for bipartite graphs?

\item The implementation of Edmonds' algorithm described in Section~\ref{sect:ImplementingEdmonds} uses an (integer)
priority queue supporting insert and delete-min, but does not take advantage of fast decrease-keys.  
Given an integer priority queue supporting $O(1)$ time decrease-key and $O(q)$ time insert and delete-min, is it possible 
to implement Edmonds' search in $O(m + nq)$ time, matching the bound for a Hungarian search~\cite{FT87,Tho03} on a bipartite graph?
\end{itemize}

\paragraph{Acknowledgement.} We would like to thank the two anonymous reviewers, whose careful reading lead to numerous improvements to the presentation.


\end{document}